\documentclass[sigconf]{aamas} 
\setcopyright{none}
\renewcommand\footnotetextcopyrightpermission[1]{}

\usepackage{balance} 
\usepackage{caption}
\usepackage{subcaption} 
\usepackage[inline]{enumitem}[leftmargin=*,noitemsep,topsep=0pt] 
\usepackage{algorithm}
\usepackage{algpseudocode}

\title{Reconstructing Network Outbreaks under Group Surveillance}

\author{Ritwick Mishra}
\affiliation{
    \institution{Department of Computer Science \& Biocomplexity Institute,\\University of Virginia,}
    \city{Charlottesville, VA}
    \country{USA}}
\email{mbc7bu@virginia.edu}

\author{Abhijin Adiga}
\affiliation{
    \institution{Biocomplexity Institute,\\University of Virginia,}
    \city{Charlottesville, VA}
    \country{USA}}
\email{abhijin@virginia.edu}

\author{Anil Vullikanti}
\affiliation{
    \institution{Department of Computer Science \& Biocomplexity Institute,\\University of Virginia,}
    \city{Charlottesville, VA}
    \country{USA}}
\email{vsakumar@virginia.edu}

\begin{abstract}
A key public health problem during an outbreak is to reconstruct the disease cascade from a partial set of confirmed infections. This has been studied extensively under the Maximum Likelihood Estimation (MLE) formulation, which reduces the problem to finding some type of Steiner subgraph on a network. 
Group surveillance like wastewater or aerosol monitoring 
is a form of mass/pooled testing where samples from multiple individuals are pooled together and tested once for all. While a single negative test clears multiple individuals, a positive test does not reveal the infected individuals in the test pool. 
We introduce the \probPoolMLE{} problem in the setting of a network propagation process, where the goal is to find a MLE cascade subgraph which is consistent with the pooled test outcomes. 
Previous work on reconstruction assumes that the test results are of individuals, i.e., pools of size one, and requires a consistent cascade to connect the positive testing nodes. 
In \probPoolMLE{}, a consistent cascade must choose at least one node in each positive pool, adding another combinatorial layer. We show that, under the Independent Cascade (IC) model, \probPoolMLE{} is NP-hard, and present an approximation algorithm based on a reduction to the Group Steiner Tree problem. 
We also consider a one-hop version of this problem, in which the disease can spread for one time step after being seeded. We show that even this restricted version is NP-hard, and develop a method using linear programming relaxation and rounding. We evaluate the performance of our methods on real and synthetic contact networks, in terms of missing infection recovery and prevalence estimation. We find that our approach outperforms meaningful baselines which correspond to pools of size one and use state-of-the-art methods.

\end{abstract}


\newcommand{\probPoolMLE}{\textsc{PoolCascadeMLE}}
\newcommand{\probPoolMLEnoisy}{\textsc{NoisyPoolCascadeMLE}}
\newcommand{\probOneHop}{\textsc{One-HopCascadeMLE}}

\newcommand{\Cost}{\text{Cost}}
\newcommand{\CostOne}{\text{Cost}^1}

\newcommand{\algo}{\textsc{ApproxCascade}}
\newcommand{\algoOne}{\textsc{RoundCascade}}
\newcommand{\erdosrenyi}{Erd\H{o}s-R\'{e}nyi}

\setlist[itemize]{leftmargin=15pt,labelsep=5pt,noitemsep,topsep=5pt}
\setlist[enumerate]{leftmargin=13pt,labelsep=3pt,noitemsep,topsep=2pt}

\newtheorem{theorem}{Theorem} 

\newtheorem{lemma}{Lemma}

\newtheorem{assumption}{Assumption}

\newtheorem{observation}[theorem]{Observation}

\renewcommand{\hat}{\widehat}

\begin{document}


\pagestyle{fancy}
\fancyhead{}


\maketitle 


\section{Introduction}
Inferring the characteristics of an outbreak from partial observations 
is an important problem in many domains such as computational 
epidemiology, biological invasions, and social 
contagions~\cite{shaman2020estimation,jang:icdm21,robinson2017invasive,mishra2023reconstructing}. 
This is a very challenging problem due to multiple reasons.
First, the number of tests that is done is typically much smaller than the population.
Second, disease transmission is very complex, and there might be many possible
outbreak scenarios consistent with observed test results.
There has been a lot of work on formalizing how to reconstruct an outbreak from available information. 
A common and intuitive approach has been to find a maximum likelihood estimation (MLE) solution. 
However, solving for the MLE is quite challenging, and several works have approximated it with a minimum cost Steiner tree, which is well understood~\cite{rozenshtein:kdd16,jang:icdm21,zhu2014information,shah:itit11}.
While this is reasonable in many settings, in general, the MLE solution can be
quite different from a Steiner tree, and~\citet{mishra2023reconstructing} show that the MLE solution can, instead, be found directly.

Group surveillance is increasingly emerging as a practical way of monitoring diseases in populations.
A classical idea in this regard is pool testing~\cite{mutesa2021pooled, finster2023welfaremaximizing}.
where the objective is to test a set of samples (or ``pool'') simultaneously.
If the pool test result is negative (and the test is assumed to be perfect), all the individuals in the pool can be cleared of the infection.
On the other hand, if a pool test result is positive, one can only infer that at least one of the individuals in the pool was positive. The main motivation is that the availability of tests is often quite limited, especially in the early stages of a pandemic, as in the case of COVID-19. In recent times, various new environmental surveillance methods like
wastewater~\cite{kilaru2023wastewater,levy2023wastewater,amman2022viral,wolfe2024detection,galani2022sars} and bioaerosol~\cite{wang2024cost,bilo2025indoor} monitoring have emerged as important group surveillance methods.

In this work, we consider the problem of reconstructing a cascade in a group surveillance setting.
Despite a lot of interest in group surveillance, this problem hasn't been considered. Prior work on pool testing for infectious diseases~\cite{mutesa2021pooled, finster2023welfaremaximizing} has focused 
on the problem of allocating tests to maximize the number of entities who get cleared (by negative tests).
Wastewater monitoring studies have addressed problems such as determination of hospitalization rates, forecasting, risk assessment,
and optimal placement of sensors~\cite{wang2024cost,calle2021optimal}.
Air sampling studies consider the problem of timely detection of bioaerosols in hostpital settings
and livestock operations~\cite{anderson2017use,bilo2025indoor}.
Note that the prior work on cascade reconstruction can be viewed as a special case of this problem with pool size~1.

Our contributions are as follows:\\
1. We introduce the \probPoolMLE{} problem under the Independent Cascade (IC) model of disease
transmission on a network as constructing an MLE solution from given group test results (Section~\ref{sec:prelim}).
We show that \probPoolMLE{} is NP-hard to approximate within a factor of $O(\log^{2-\epsilon}{k})$, where $k$ denotes the number of node groups in the problem instance.
In contrast, when the group size is 1, an $O(\log{k})$ approximation is possible~\cite{mishra2023reconstructing}.
Here, we present a $O(k^\epsilon)$-approximation algorithm called \algo{}.
We also consider the more general setting where tests are unreliable, and show that a similar approximation is possible.\\
2. We consider a special case of the MLE estimation problem, \probOneHop{}, which corresponds
to for one step of disease spread.
This is specially motivated by group surveillance in wastewater and animal farms.
We show that this problem is also NP-hard, and give an $O(\log{k})$ approximation algorithm.\\
3. We evaluate the performance of \algo{} on two tasks: (a) recovering the infected nodes, (b) estimating prevalence, i.e., outbreak size. We use synthetic data generated on multiple synthetic and real contact networks, including
a contact network built from the Electronic Health Record (EHR) for the UVA Hospital ICU and a large realistic contact network representing the population of a small city. We find \algo{} compares favorably to baselines which reduce pool size to 1 and use the state-of-the-art method in \cite{mishra2023reconstructing}. We also evaluate \algoOne{} against a random baseline and find that it too outperforms in missing infection recovery. \\
4 . We examine the conditions in which the \probPoolMLE{} solution does not recover the ground truth to show the limitations of the MLE approach.    
We also show that even a little bit of noise in the test results can change the MLE solution of \probPoolMLEnoisy{} quite significantly, compared to that of  \probPoolMLE{}.

\section{Related work}
In the pool testing literature, the general goal is to identify the subset
of infections in a relatively very large population by simultaneously
testing multiple individuals. The idea was first proposed during World
War~II in the context of detecting syphilis infected
population~\cite{dorfman1943initialpools}. Subsequently, it has been
applied in various domains such as industrial testing, experiment design,
coding theory~\cite{du2000combinatorial}. Multiple variants of the problem
exist accounting for aspects such as adaptive or non-adaptive setting,
noiseless or noisy tests, and exact or partial
recovery~\cite{aldridge2019group}. There has been renewed interest in this
topic in the context of COVID-19~\cite{sunjaya2020pooled,mutesa2021pooled,finster2023welfaremaximizing}. In recent years,
there has been work that accounts for heterogeneous interactions of
individuals in the population in the design of group testing algorithms.
 References
\cite{sewell2022leveraging,arasli2023group} assume the knowledge of the
network of interactions for better pool design. To the best of our
knowledge, existing works focus on designing the pools with the objective
of minimizing the number of tests to be performed in both non-adaptive and
adaptive settings.  

Our work is motivated by the emergence of various environmental surveillance methods
like wastewater and bioaerosol monitoring in the context of infectious diseases
in humans and animals. 
The objective is to estimate important characteristics of the disease from 
observations for downstream tasks of risk assessment, situational awareness, and 
forecasting~\cite{chen2024wastewater}. In this setting, pools are already provided 
along with test results (see for example the national-scale monitoring of H5N1~\cite{louis2024wastewater,honein2024challenges}). 

In the setting of network propagation processes, reconstructing a cascade
given partial information about infected nodes (pool size~1) is
well-studied. \citet{rozenshtein:kdd16} use a directed Steiner tree
approach in a temporal network setting, where a subset of infected nodes
along with the time of infection is provided. \citet{jang:icdm21} account
for node attributes to detect asymptotic cases.
\citet{mishra2023reconstructing} formulate an MLE problem to reconstruct
cascade given a subset of observed infected nodes and diffusion model.
Unlike previous works, they consider true MLE cost that includes failed
infection attempts, and use a node-weighted Steiner tree approach for
cascade reconstruction. Our work uses this MLE formulation, but generalizes
this work to the pool-testing scenario. \citet{qiu2023reconstructing}
consider the problem of reconstructing diffusion history from a single
snapshot of the cascade, i.e., the knowledge of all infected nodes. Their
barycenter formulation does not assume the knowledge of the diffusion
parameters.

\section{Preliminaries}
\label{sec:prelim}


\textbf{Disease model. }
We consider the Independent Cascade (IC) model of disease spread which is the simplest form of the discrete-time SIR model, on an undirected network $G=(V,E)$. 
Each node is in one of the following states: susceptible (S), infectious (I) or removed (R). Let $V_0\subset V$ denote the initial set of nodes in state I that seed the disease at time $t=0$. At any time $t$, each node in state I infects its susceptible neighbor with probability according to the weight on the respective edge. Then, $v$ transitions to state R at time $t$ , thus getting only one chance to infect.
We consider two variations of this IC process:\\ 
(a) \texttt{single-seed}: The disease begins from a single known seed node, i.e. $V_0 = \{s\}$, and can spread over multiple time steps, until there are no more new infections. \\
(b) \texttt{one-hop}: The disease spread is limited to one time step after seeding. Here $V_0$ is not known.\\
One realization of this random process is called a \emph{cascade}. 



\noindent
\textbf{Pool testing.}  
A pool test  corresponds to a subset $g_i\subset V$.
This indicates that the subset $g_i$ are tested simultaneously; the test result is positive if any node in $g_i$ is infected.
A set $\Gamma=\{g_1 , g_2, \dots, g_k\}$, with each $g_i\subset V$ denotes a group level surveillance strategy.
We use $\Gamma_1$ to denote the set of pools which test positive (i.e., at least one node amont them is infected); $\Gamma_0=\Gamma - \Gamma_1$ is the set of pools which test negative.

\noindent
\textbf{Criteria for a consistent cascade. } Given a set of observations $(\Gamma_0, \Gamma_1)$, we require for a cascade $A$ to be consistent that it meet both of these criteria: (a) that it does not include any node from $\Gamma_0$, i.e, $\bigcup_{g\in\Gamma_0}g \cap V(A) = \emptyset$, and  (b) that it includes at least one node from each pool in $\Gamma_1$, i.e, $\forall g\in \Gamma_1, V(A) \cap g \neq \phi$. 

\noindent
\textbf{Example.}
Figure \ref{fig:pool_recon_ex} (left) shows a cascade (red edges) on a graph with 10 nodes.
There are two pool tests $\Gamma = \{g_1=\{5, 6, 7\}, g_2 = \{3, 8, 9\}\}$.
Suppose nodes 7 and 3 are infected.
Then, both the pools $g_1$ and $g_2$ will test positive, and $\Gamma_1=\{g_1, g_2\}$.
$T_r$ (red edges) denotes a disease cascade with root $r$.
The subgraph in purple edges in Figure \ref{fig:pool_recon_ex}(right) denotes a consistent cascade, since it ensures that both positive groups have a node connected to $r$.
The inferred cascade in the right is not the same as the one on the left, but shares some nodes.

\subsection{Problem formulation: MLE construction for \texttt{single-seed} instance}

\begin{figure}[tb]
    \centering
    \includegraphics[width=1.0\columnwidth]{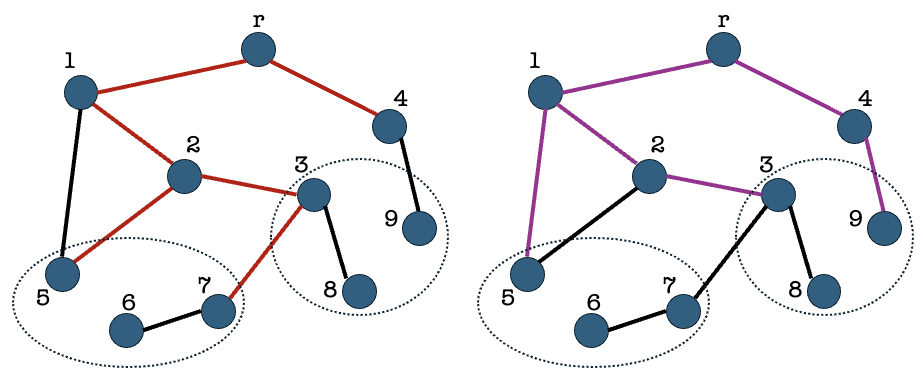}
    \caption{On the left, a cascade (in red) with root $r$ has resulted in two positive-testing pools (in dashed ovals). On the right, a reconstructed cascade $T_r$ is shown (in purple) which is \emph{consistent}, i.e., it contains at least one node from each positive pool. Here, $\delta_{T_r} = \{(3,7),(3,8)\},\lambda_{T_r}=\{(2,5)\}$. The probability $P(T_r)= p_{r1}p_{15}p_{12}p_{23}p_{r4}p_{49}(1-p_{37})(1-p_{38})$.}
    \label{fig:pool_recon_ex}
\end{figure}


We first define the probability of a cascade, and transform that to a cost, in order to formalize the MLE problem.
Let the set of all nodes in the negative pools be $S_0 = \bigcup_{g\in \Gamma_0} g$. Let $N_e(u)$ be the set of edges incident at node $u$.
Given a cascade $T_r$ in $G=(V,E)$,
we define:
(1) $E(T_r)$: the set of edges in $T_r$;
(2) $E(T_r,S_0)$: the set of all edges with one endpoint in $T_r$ and the other in $S_0$;
(3) $\delta_{T_r}$: the set of edges with one endpoint in $T_r$ and the other in $V\setminus (V(T_r)\cup S_0)$;
(4) $\lambda_{T_r}$: the set of edges with both endpoints in $T_r$, but not belonging to $T_r$.
Let $d_{T_r}(u,v)$ denote the distance in $T_r$ between
the nodes $u$ and $v$.
Under the IC disease model , the probability of the cascade $T_r$ is:
\begin{align*}
    \overline{P}(T_r) &= \prod_{e \in E(T_r)} p_e  
    \prod_{e\in E(T_r, S_0)} ( 1 - p_e)
    \prod_{e\in \delta_{T_r}} (1 - p_e) \\
      & \quad \prod_{\substack{e =(u,v)\in \lambda_{T_r}, \\ d_{T_r}(r,u) \neq d_{T_r}(r,v)}} (1 - p_e)
\end{align*}

\noindent
\textbf{Example.}
In the example in Figure \ref{fig:pool_recon_ex} (right), 
the set $E(T_r)$ for the cascade $T_r$ consists of the purple edges.


An MLE cascade is one which maximizes the above probability. 
Let $c_e = -\log{p_e}$ and $d_e = -\log{(1-p_e)}$.
As in prior formulations of MLE~\cite{mishra2023reconstructing}, we formalize the cost of a cascade as 
\begin{align}
    \Cost(T) = \sum_{e\in E(T)}c_e 
    + \sum_{e \in E(T,S_0)} d_e 
    + \sum_{e \in \delta_T} d_e
    + \sum_{e \in \lambda_T} d_e 
    \label{eq:approxCost}
\end{align}
As in~\cite{mishra2023reconstructing}, we disregard the fact that two neighbors in the cascade
equidistant from the source (in the cascade) do not contribute $d_e$ costs.




\noindent
\textbf{The \probPoolMLE{} problem. }Given an undirected graph $G=(V,E)$, a seed $r$, a set of pool-tested node groups $\Gamma=\Gamma_0\cup \Gamma_1$, 
find a subgraph $T_r$ rooted at $r$, which is consistent with $\Gamma$ and minimizes $\Cost(T_r)$.

We also consider a version where the tests are imperfect, i.e., the observations are noisy.

\noindent
\textbf{The \probPoolMLEnoisy{} problem. } Given an undirected graph $G=(V,E)$, a set of noisy pool-tested node groups $\Gamma=\Gamma_0\cup \Gamma_1$, find outcomes $\Gamma'_0\cup\Gamma'_1$, and a subgraph $T$, which is consistent with $\Gamma'_0\cup\Gamma'_1$, and minimizes $\Cost(T|\Gamma'_0\cup\Gamma'_1) + \log{(1/q(\Gamma_0, \Gamma_1, \Gamma'_0, \Gamma'_1))}$. Here $q(\Gamma_0, \Gamma_1, \Gamma'_0, \Gamma'_1)$ denotes the probability that the actual outcomes are $(\Gamma_0',\Gamma_1')$ when the observed outcomes are $(\Gamma_0,\Gamma_1)$.



\subsection{Problem formulation: MLE construction for the \texttt{one-hop} instance}

For simplicity, we assume the disease starts on nodes in $V_0$ and spreads to nodes in $V_1=V\setminus V_0$, and we are given observations from $V_1$.
For a cascade consisting of a subset $A$ of edges, let $V_0(A)$ denote the set of seeds in $V_0$.
Let $\sigma$ denote the set of edges between $V_0$ and $V_1$.
We can write the probability of cascade $A$ in terms of the probabilities of the constituent seeding and transmission events as follows.
\begin{align*}
    P(A)  = \prod_{v\in V_0(A)} p_v^{0} \prod_{v\in V\setminus V_0(A)} (1-p_v^0) \prod_{e\in E(A)}p_e \prod_{e\in \sigma_A \setminus E(A)}(1-p_e)
\end{align*}
Equivalently, we can write the objective in terms of costs. In addition to $c_e, d_e$, we define $a_v = -\log p_v^0$ as the cost of seeding $v$, and $b_v=-\log (1-p_v^0)$ as the cost of not seeding $v$. The cost of a cascade under the one-hop model is,
\begin{align}
    \CostOne(A) = \sum_{v\in V_0(A)} a_v + \sum_{v\in V\setminus V_0(A)}b_v + \sum_{e\in E(A)}c_e + \sum_{e\in \sigma_A\setminus E(A)}d_e
\end{align}
\textbf{The \probOneHop{} problem. } Given a set of observations $\Gamma= (\Gamma_0,\Gamma_1)$ , find a cascade $A$ which is consistent and minimizes $\CostOne{}(A)$.

\section{Hardness results}
\label{sec:hardness}

We show that both \probPoolMLE{} and \probOneHop{} are computationally hard, even to approximate.
In contrast, the MLE for the problem without pool testing (equivalently, all pools of size 1), which was considered in~\cite{mishra2023reconstructing}, can be approximated within an $O(\log{n})$ factor.

\begin{theorem}\label{thm:poolmlehardness}
\probPoolMLE{} is hard to approximate within a $O(\log^{2-\epsilon}{k})$ factor, for any $\epsilon>0$ unless P=NP. Here $|\Gamma_1|=k$.
\end{theorem}
\noindent
The proof is in the supplement. It is a reduction from
the Group Steiner Tree problem.

\begin{theorem}\label{thm:onehophardness}
\probOneHop{} is NP-hard to approximate within a $O(\log{k})$ factor, unless P=NP.
\end{theorem}
\noindent
The proof is in the supplement. It is a reduction from
the Minimum Set Cover problem.

\section{Our Approach}

Due to the significant hardness of the  \probPoolMLE{} and \probOneHop{} problems (Section~\ref{sec:hardness}), we focus on approximation algorithms here. We make the following natural assumption about the disease regime.

\begin{assumption} \label{prob_assume} The edge transmission probability $p_e \leq 1/2$ for all edges, i.e.,~$c(e) \geq d(e) $,  for all $e\in E$.
\end{assumption}

\subsection{\probPoolMLE{} problem}

\begin{lemma}
Under Assumption \ref{prob_assume}, a \probPoolMLE{} solution $T^\star_{r^*}$ is a tree.
\end{lemma}
\noindent
If the solution has a cycle, then we can reduce its cost by removing an edge in the cycle and still have a feasible solution. Hence, the solution is a tree.

\noindent
Our approach is based on a reduction to the Group Steiner Tree (GST) problem on a graph with  node and edge weights. We are given an undirected graph $G$ and a set of pools $\Gamma$. We construct graph $G'$ from $G$ by adding node and edge weights as defined in 
Algorithm~\ref{alg:approx_cascade}, \textsc{ApproxCascade}. We remove the set of nodes belonging to the negative pools, $S_0$. Now our goal is to find a minimum weighted tree in $G'$ that connects to at least one node in each of the positive pools, $\Gamma_1$, which is the Group Steiner Tree problem \cite{ihler1991Bounds, charikar1999}.
Lemma \ref{lem:pool2gst} gives the bounds on the weight of such a tree in terms of our \Cost{} (Expression (\ref{eq:approxCost})). Let $N_e(u)$ be the set of all incident edges on node $u$.

\begin{lemma}
\label{lem:pool2gst}
    Let $T$ be any Group Steiner Tree in $G'$ with respect to the groups $\Gamma_1$. Let the weight of $T$ in $G'$ be denoted as $w(T)$. Then, $T$ is consistent with $\Gamma$ and 
    \begin{align*}
        \Cost(T) \leq w(T) \leq 2\ \Cost(T)
    \end{align*}
\end{lemma}

\begin{proof}
    For any group Steiner tree $T$ on $G'$, we have
\begin{align*}
    & w(T) = \sum_{u\in V(T)} w(u) + \sum_{e \in E(T)} w(e) \\
    &= \sum_{u \in V(T)} \sum_{e \in N_e(u)} d_e + \sum_{e \in E(T)}(c_e - d_e) \\
    &= \sum_{u\in V(T)}\Big[\sum_{e \in N_e(u) \cap E(T) } d_e + \sum_{e \in N_e(u) \cap E(T,S_0)} d_e\\
    & \quad
    +\sum_{e \in N_e(u) \cap \delta_T} d_e  + \sum_{e \in N_e(u) \cap \lambda_T}d_e \Big] + \sum_{e \in E(T)}(c_e - d_e) \\
    &= 2\sum_{e \in E(T)}d_e + \sum_{e\in E(T,S_0)} d_e+ \sum_{e\in \delta_T} d_e + 2\sum_{e\in \lambda_T} d_e \\
    & \quad + \sum_{e \in E(T)}(c_e - d_e) \\
    &= \Cost(T) + \sum_{e \in E(T)} d_e + \sum_{e\in \lambda_T}d_e \\
    &\leq 2~\Cost(T) - \sum_{e \in E(T,S_0)}d_e - \sum_{e \in \delta_T} d_e \\
    &\Rightarrow w(T) \leq 2~\Cost(T)\,.
\end{align*}
\end{proof}

\begin{lemma}
    \label{lem:nwgst2ewdst}
    (Charikar~et~al.\cite{charikar1999}) There is an approximation preserving reduction from the Group Steiner Tree problem on a node- and edge-weighted graph to the  Directed Steiner Tree problem on a purely edge-weighted graph. 
\end{lemma}

\noindent
We describe the reduction in Lemma \ref{lem:nwgst2ewdst} in the supplement (Algorithm~\ref{alg:GST}). We convert the original graph into a purely edge-weighted graph by employing `in' and `out' copies of nodes in the usual manner. Then, we add a dummy node for each group, connecting each node in the group to this dummy node with zero-weighted edges. With the dummy nodes as terminals, we solve the Directed Steiner Tree problem~\cite{charikar1999}. Charikar~et~al.~\cite{charikar1999} provide the best-known approximation algorithm for the Directed Steiner Tree problem which has an approximation ratio of $O(k^\epsilon)$ and runs in time $O(kn^{1/\epsilon})$ for a fixed $\epsilon>0$, where $k$ is the number of terminals and $n$ is the size of the network. This leads to an approximation ratio of $O(k^\epsilon)$ for our algorithm, \textsc{ApproxCascade}. The proof is in the supplement. 

\begin{theorem}
\label{theorem:approx_bound}
    Let $\hat{T}_r$, rooted at $r$, be the tree returned by the Algorithm \ref{alg:approx_cascade}.
    Let $T^*_{r^*}$ be an optimal solution to \probPoolMLE{}, rooted at $r^*$. Then,
    \begin{align*}
        \Cost(\hat{T}_r) \leq O(k^\epsilon)\Cost(T^*_{r^*}),
    \end{align*}
    where $k$ is the number of positive pools, i.e., $|\Gamma_1|$ and a fixed $\epsilon>0$.
\end{theorem}

    



\noindent
\textbf{Time complexity. }
 Algorithm \ref{alg:approx_cascade} runs in $O(n^2 + kn^{1/\epsilon})$time. At a fixed $\epsilon = 0.5$, the time complexity is $O(kn^2)$.


\begin{algorithm}[ht]
\caption{\textsc{ApproxCascade}}
\label{alg:approx_cascade}
\raggedright{
\textbf{Input}: $G=(V,E)$, seed $r$, edge probabilities $p_e$, a set of pool-tested node groups $\Gamma=\Gamma_0 \cup \Gamma_1$. \\
\textbf{Output}: A tree $T_r$ rooted at $r$ and consistent with $\Gamma$.
}
\begin{algorithmic}[1] 
\For{each edge $e$}
\State Compute the cost of inclusion~$c_e = -\log{p_e}$ and cost of exclusion~$d_e=- \log{(1-p_e)}$ 

\EndFor
\State Construct a node and edge-weighted graph $G'$ such that,
\For{each node $u\in G$}
    \State $w(u) \gets \sum_{e \in \mathcal{N}_e(u) } d_{e}$
\EndFor
\For{each edge $e\in G$}
    \State $w(e) \gets c_e-d_e$
\EndFor
\State Remove nodes $\bigcup_{g\in \Gamma_0}g $ from $G'$.
\State $T_r =$ \textsc{GroupSteinerTree}$(G', r, \Gamma_1)$ 
\State \textbf{return $T_r$}
\end{algorithmic}
\end{algorithm}

\subsection{\probOneHop{} problem}


Let $x_i$, $y_{ij}$ and $z_{ij}$ denote indicator variables for node $i\in V_0$ being seeded, disease transmission on edge $(i, j)$ and no disease spread on edge $(i, j)$, respectively.
We first describe an integer program, which doesn't exactly solve \probOneHop{}, but is a 2-approximation.

\begin{align*}
     \text{minimize} \quad & \sum_{i\in V_0} ( a_ix_i + b_i(1-x_i)) + \sum_{(i,j)\in E} ( c_{ij}y_{ij} + d_{ij}(z_{ij}) ) \\
    \text{subject to} \quad &    \sum_{j\in g}\left[  \sum_{i:(i,j)\in E} y_{ij}\right]  \geq 1,  \quad g\in \Gamma_1 \\
    & x_i - y_{ij} \geq  0 , \quad (i,j)\in E \\
    & z_{ij} \geq x_i, \quad (i,j) \in E\\
    & x_i \in \{0,1\} , \quad i \in V_0 \\
    & y_{ij} \in \{0,1\} , \quad (i,j) \in E
\end{align*}

\noindent
The first set of constraints corresponds to the requirement that in each group $g\in \Gamma_1$, at least one node should  be covered by a live-edge. The second set of constraints ensures that, a live-edge must begin at a seeded node. 
The third set of constraints incorporates the cost of non-infection.


\begin{lemma}
Let  $x, y, z$ denote the optimal integral solutions to the above program.
Then, $\sum_{i\in V_0} ( a_ix_i + b_i(1-x_i)) + \sum_{(i,j)\in E} ( c_{ij}y_{ij} + d_{ij}(z_{ij}) )\leq 2OPT$, where $OPT$ denotes the cost of the optimal solution to the instance of \probOneHop{}.
\end{lemma}
\noindent
The proof is in the supplement.


\noindent
\textbf{LP relaxation. } We relax the integrality constraints by replacing them with $x_i, y_{ij} \geq 0, \forall i\in V, (i,j)\in E$. Let $\{x_i^*\}_{i\in V}, \{y_{ij}^*\}_{(i,j)\in E}$ be the optimal LP solution. 

\begin{algorithm}
    \caption{\textsc{RoundCascade}}
    \label{alg:round_cascade}
    \raggedright{
    \textbf{Input: } A \probOneHop{} instance\\
    \textbf{Output: } A cascade $A$ consisting of set $\{X_i: i\in V_0\}$ and edges $\{Y_{ij}: (i, j)\in E\}$
    }
    \begin{algorithmic}[1]
        \State Solve the LP to get $\{x_i^*\}_{i\in V}, \{y_{ij}^*\}_{(i,j)\in E}$.
        \For{each $i\in V$}
        \State Independently pick a number $\tau_i\in [0,1]$ uniformly at random.
        \State Set $X_i=1$ if $\alpha x_i^*>\tau_i$, otherwise set $x_i=0$.
        \EndFor
        \For{each $(i,j) \in E$}
            \If{$\alpha y_{ij}^*>\tau_i$}
                \State Set $Y_{ij}=1, X_i = 1$, $Z_{ij}=0$.
            \Else
\State Set $Y_{ij}=0$
and $Z_{ij} = X_i$ 
            \EndIf
        \EndFor
    \end{algorithmic}
\end{algorithm}




\begin{lemma}
Let $X, Y, Z$ denote the solution output by \textsc{RoundCascade}.    Then, $E[\sum_i (a_i-b_i) X_i + \sum_{ij} c_{ij} Y_{ij} + d_{ij} Z_{ij}] \leq \alpha~OPT^*$ where $OPT^*$ is the optimal objective value of the LP, and $\alpha = 1+\ln |\Gamma_1|$.
\end{lemma}
This follows directly from the randomized rounding.
Finally, we show that the solution is feasible.

\begin{lemma}
\label{lem:prob_bound}
Let $A$ denote the solution computed by \textsc{RoundCascade}.
Then, $Pr[A \mbox{ is infeasible}] \leq k e^{-\alpha} \leq 1/n^c$ for $k=\Omega(\log{n})$.
\end{lemma}
\begin{proof}
We show that for any $g\in \Gamma_1$, $\Pr[g $ is not connected by $ A]
    \leq e^{-\alpha}$.
    We have,
    \begin{align*}
        &\Pr[g \text{ is not connected by } A] 
        =\Pr[\bigcap_{j\in g}(j \text{ not in A})] \\
        &= \prod_{j\in g}\Pr[(j\text{ is not seeded})\cap (j \text{ is not covered by a live-edge})] \\
        &=\prod_{j\in g}\Pr[j\text{ is not seeded}]\Pr[j \text{ is not covered by a live-edge})] \\
        &=\prod_{j\in g}\left[ (1-\alpha x_j^*)\prod_{i:(i,j)\in E} (1-\alpha y_{ij}^*) \right]
        \leq \prod_{j\in g} \left[e^{-\alpha x_j^*}\prod_{i:(i,j)\in E}e^{-\alpha y_{ij}^*} \right]\\
        &= e^{-\alpha\sum_{j\in g}\left(x_j^* + \sum_{i:(i,j)\in E}y_{ij}^*\right)} \leq e^{-\alpha} 
    \end{align*}
Above, we have used the independence of events, $\{j \text{ not in } A\}_{j\in g}$, and of the events, ($j$ is not seeded) and ($j$ is not covered by a live-edge).
In the last line, we have used the first set of constraints in the LP.

Therefore,  $Pr[A \mbox{ is infeasible}]= Pr[\bigcup_{g\in \Gamma_1}\text{ g not connected by A}]  \leq \prod_{j\in X}\Pr[\text{ j is not in A }]\leq ke^{-\alpha}$ by union bound.
\end{proof}

\begin{theorem}
    \textsc{RoundCascade} is a randomized $(2+2\ln k)$-approximation algorithm for \probOneHop{}, where $|\Gamma_1|=k$.
\end{theorem}

\section{Experimental Results}
\label{sec:exp}
\begin{table*}[h]
    \caption{Networks and their properties. Shortest path length for the small-city network is omitted as its relevance depends on the diffusion rate~$\beta$.}
    \label{tab:datasets}
\begin{center}
\begin{normalsize}
\begin{tabular}{|c|p{0.6in}|p{0.6in}|p{0.6in}|p{0.6in}|p{1in}|}
    \hline
    \textbf{Graph Name} & \textbf{Nodes} & \textbf{Edges} & \textbf{Clustering coefficient} & \textbf{Avg. shortest path length} & \textbf{Note}\\
    \hline \hline 
    BA $m=3$ & 1000 & 2991 & 0.036 & 3.477 & Synthetic network\\
    \hline 
    $G(n,q)$ random graph & 1000 & 9974 & 0.0197 & 2.6398 & Synthetic network\\
    \hline
   \texttt{hospital-icu} & 879 & 3575 & 0.59 & 4.31 & Real-world network\\
    \hline
   \texttt{small-city} & 10001 & 52575 & 0.349 & -- & Realistic synthetic weighted network.\\
    \hline
    \hline
\end{tabular}\\
\end{normalsize}
\end{center}
\end{table*}
\subsection{Dataset and Methods}
We experimentally evaluate the performance of \textsc{ApproxCascade} and \algoOne{} 
using the networks listed in Table \ref{tab:datasets} and described below.
\begin{enumerate}
\item
\texttt{Barabasi-Albert (BA) networks}: Since many real-world networks have been observed to be in the class of scale-free networks, BA networks are a useful model. We generate random BA networks with the number of edges for each new node, $m=3$.
\item
\texttt{\erdosrenyi{} graphs}: We generate a $G(n=1000,q=0.02)$ graph for evaluating \textsc{ApproxCascade}. 
\item
\texttt{hospital-icu}:
This is a contact network between patients and healthcare workers built from Electronic Heath Records (EHR) for the UVA Hospital ICU for the time period between Jan 1, 2018 to Jan 8, 2018.
\item
\texttt{small-city}:
This is a subgraph of the Virginia~(VA) digital twin-based network made available at ~\cite{virginiadata} and used previously by~\cite{chen2025epihiper,hoops2021high}, created on a synthetic population based on census data and community surveys. Each edge $e$ has a contact duration $d_e$, which is used to generate the edge-diffusion probability $p_e=1-e^{-\beta d_e}$; $\beta$ is a parameter to control transmissibility.

\end{enumerate}

\noindent
\textbf{Baselines.} Since this problem has not been considered before, we design the following baselines.
\begin{enumerate}
    \item \textsc{ApproxCascade-Random} and \textsc{RoundCascade-Random}: Instead of choosing the infected nodes in a positive pool cost-effectively, randomly choose one node to be the only infection in that pool, reducing to a pool size of 1. 
    \item \textsc{ApproxCascade-All}: Here, we consider every node in a positive pool to be infected, also reducing to pool sizes of 1.
\end{enumerate}
Both these baselines reduce the pool size to~1 which can now be solved by 
the state-of-the-art method in \cite{mishra2023reconstructing}. For speed
and scalability, we deploy Charikar~et.~al.'s node-weighted Steiner tree 
solver as required by that method~\cite{charikar1999}. Our goal is to 
understand the conditions where it is important to systematically select the likely infections within a pool.

\noindent
\textbf{Method.} 
 We randomly select a fixed proportion of nodes in a given network to be pooled according to a parameter, \emph{pool ratio} set to $\{0.5, 0.9\}$. We randomly bunch these selected nodes into equal-sized pools whose size is determined by a parameter, \emph{pool size} which can be $\{3, 5, 7, 9\}$. 
 \underline{\algo{}:}
 We generate IC simulations starting at a random seed and determine the subset of positive pools which are pools having at least one infected node.
\underline{\algoOne{}:} We first construct a time-expanded version of the network, $G' =((V'_0,V'_1), E')$, where each node in $u \in V$ corresponds to two timestamped copies in $u_0\in V'_0, u_1 \in V'_1$, and each undirected edge $(u,v)$ corresponds to two copies $(u_0, v_1),(v_0, u_1)\in E'$. On $G'$, which is an oriented bipartite network, we assume the pooled observations come only from the $V'_1$ part. To generate infections, we randomly choose seeds from $V'_0$ with homogenous probability $p^0 \in \{0.01, 0.05, 0.10\}$, starting from which we simulate for one time-step.
For the unweighted networks, we use homogeneous diffusion probabilities $p=\{0.01, 0.05, 0.10, 0.20\}$ and for the weighted network, we use the parameter $\beta=\{2,3,5,7\} \times 10^{-6}$. We give the network, the disease parameters, the positive and negative pools (and in case of \algo{}, the seed) as input to our algorithms as well as the baselines.

\noindent
\textbf{Performance measures.} To evaluate the performance of our method, we have two tasks: 
\begin{enumerate}
\item \textit{Missing infection recovery}: We compare the node set of the reconstructed subgraph with the ground truth infection set. As metrics, we choose F1-score to quantify the success in missing infection recovery, as used in previous works~\cite{rozenshtein:kdd16,jang:icdm21,mishra2023reconstructing}. These are defined as follows.
Let the node-set of the reconstructed cascade $V_T$ and the ground truth infected node-set $V_G$. Define true positives as $TP = |V_T \cap V_G|$, true negatives as $TN = |V \setminus (V_T\cup V_G)|$, false positives as $FP = |V_T \setminus V_G|$, and false negatives as $FN=|V_G \setminus V_T|$. We compute F1-score $= {2TP}/{(2TP + FP + FN)}$.
\item \textit{Prevalence estimation}: We compare the sizes of the reconstructed cascade and the ground truth cascade using the relative error $e_{rel} = (|V_G|-|V_T|)/|V_G|$.
\end{enumerate}
The presented results are averaged over 50 replicate runs.




\subsection{Results}
\begin{figure*}[!htb]
    \centering
    \begin{subfigure}{\textwidth}
    \centering
    \includegraphics[width=.28\textwidth]{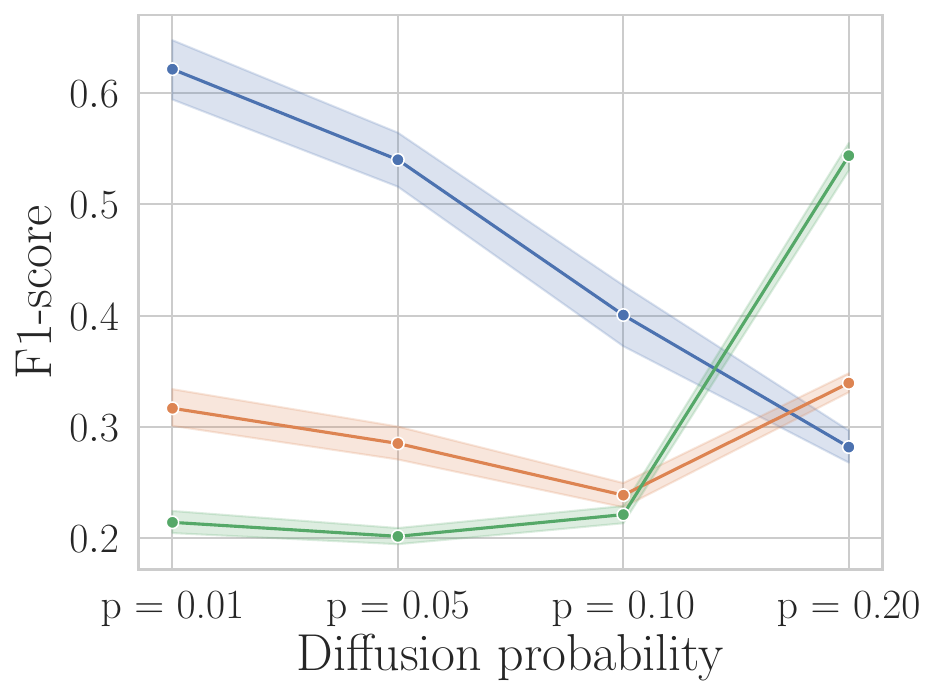}
    \includegraphics[width=.28\textwidth]{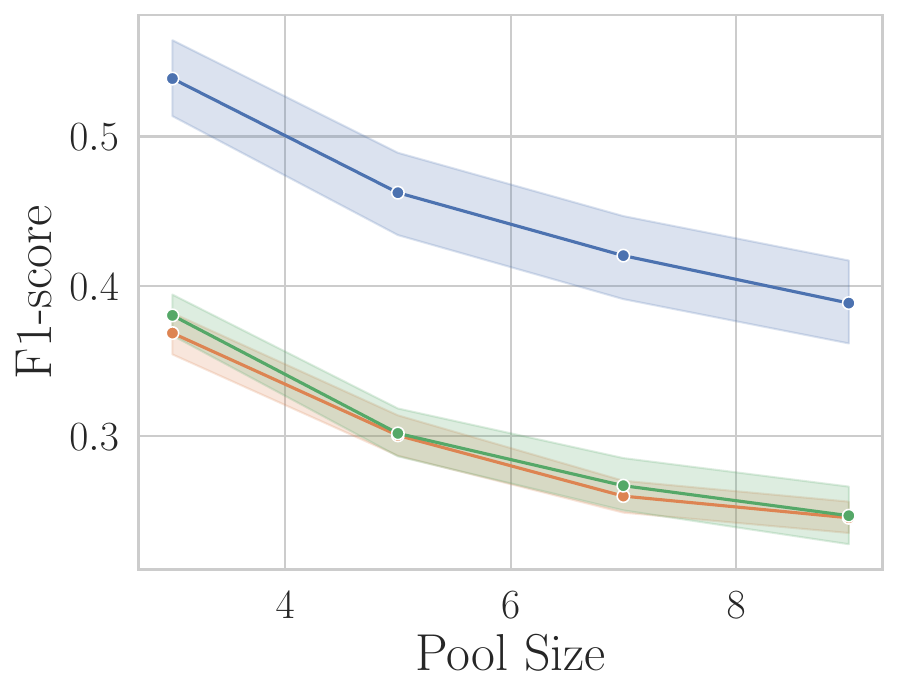}
    \includegraphics[width=.28\textwidth]{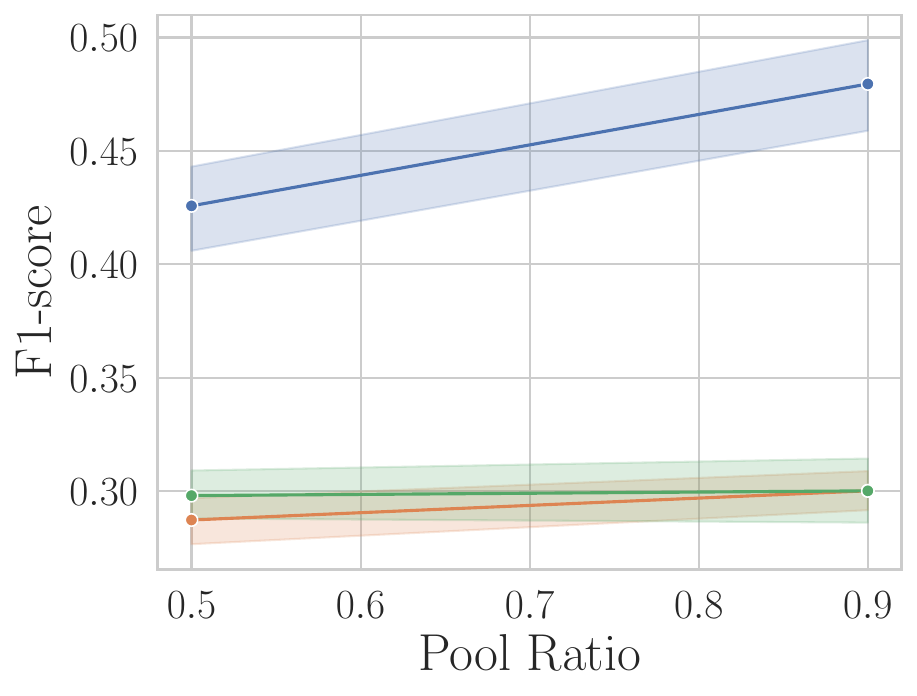}
    \caption{BA $m=3$}
    \end{subfigure}
    \begin{subfigure}{\textwidth}
    \centering
    \includegraphics[width=.28\textwidth]{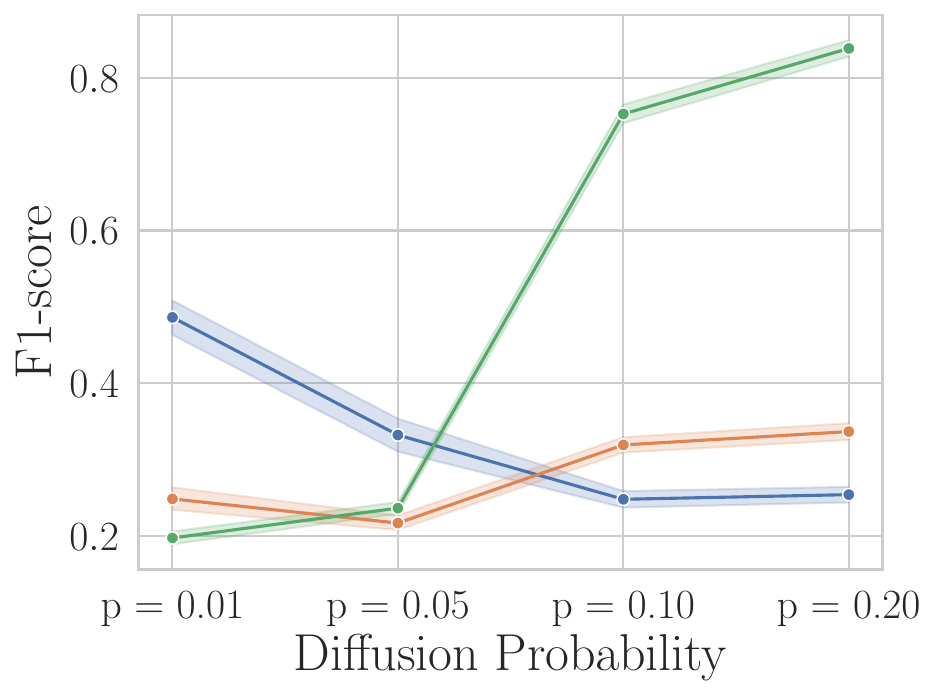}
    \includegraphics[width=.28\textwidth]{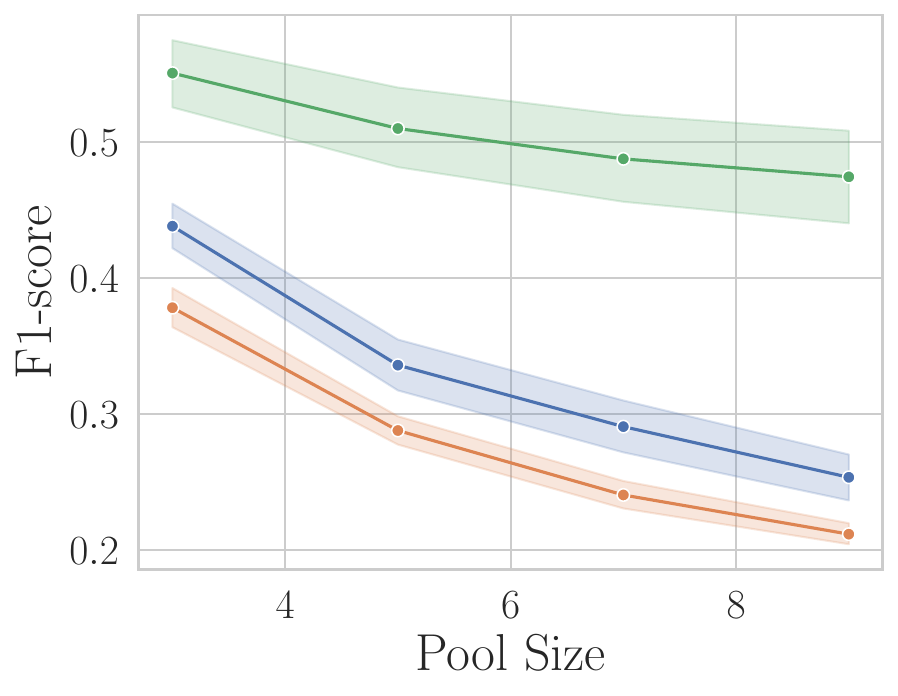}
    \includegraphics[width=.28\textwidth]{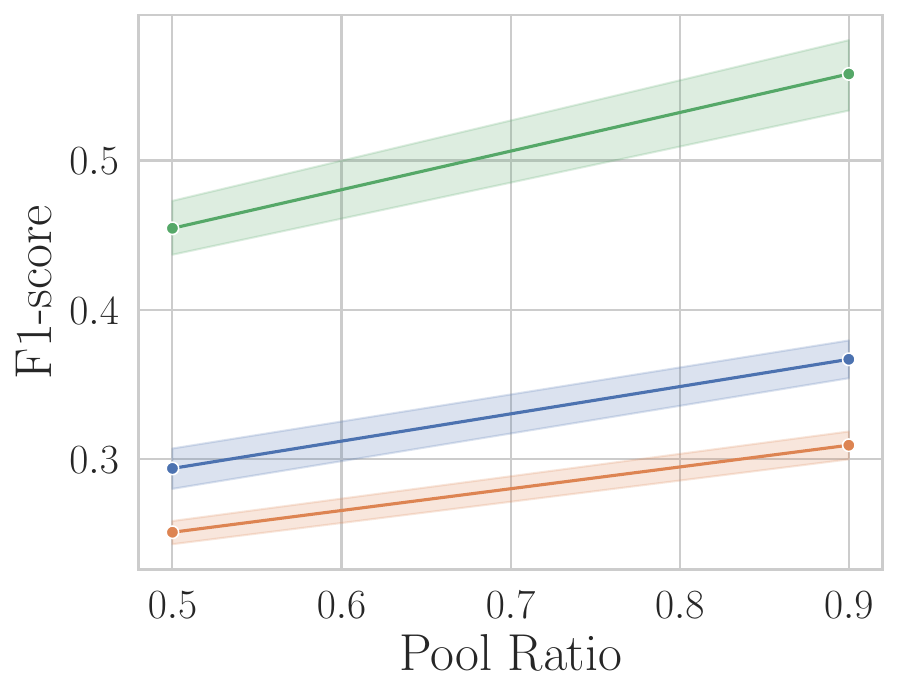}
    \caption{$G(n=1000,q=0.02)$}
    \end{subfigure}
    \begin{subfigure}{\textwidth}
    \centering
    \includegraphics[width=.28\textwidth]{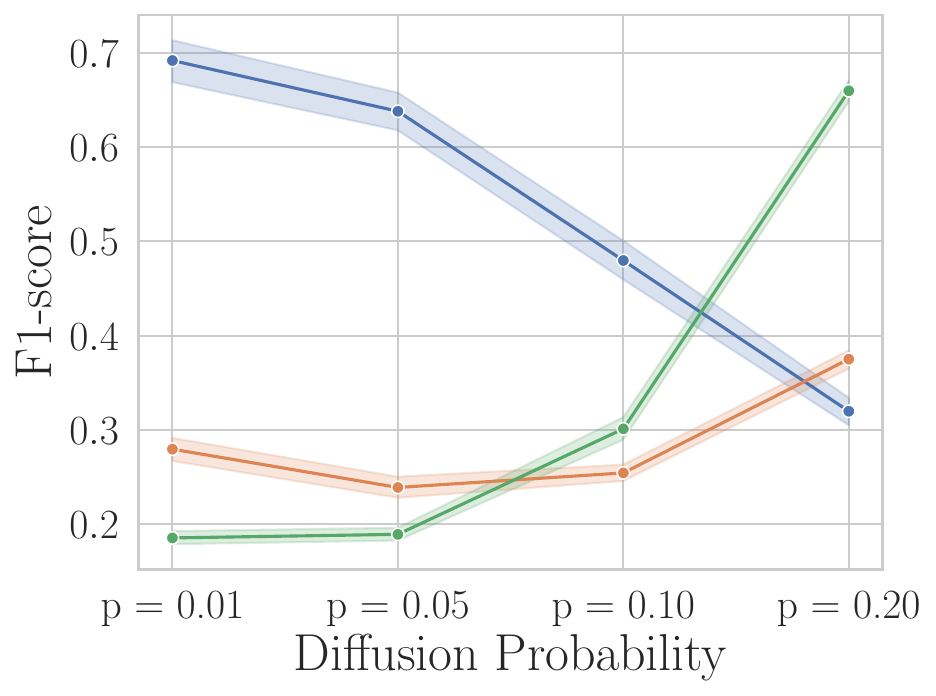}
    \includegraphics[width=.28\textwidth]{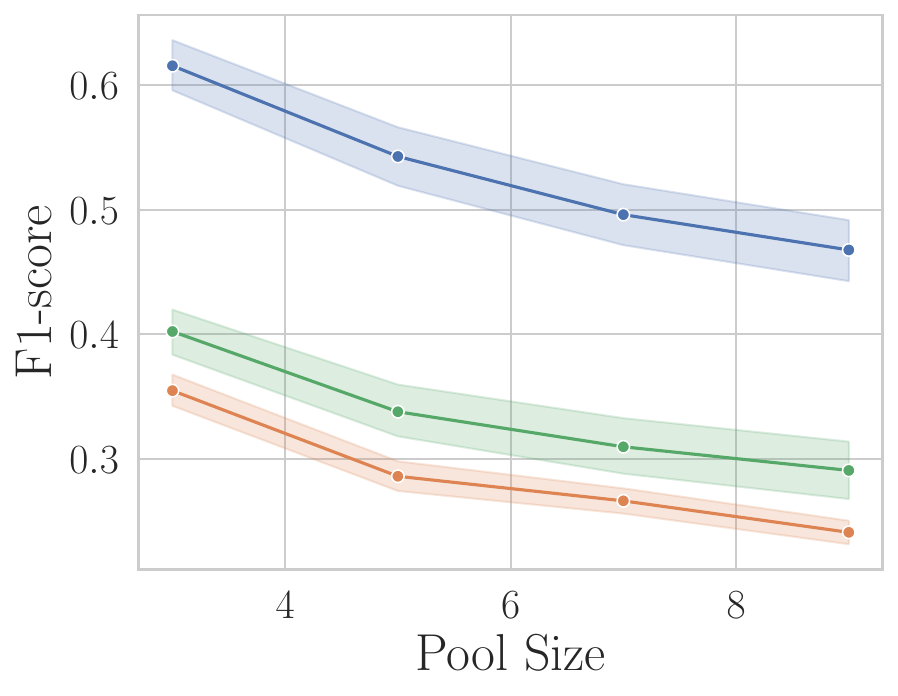}
    \includegraphics[width=.28\textwidth]{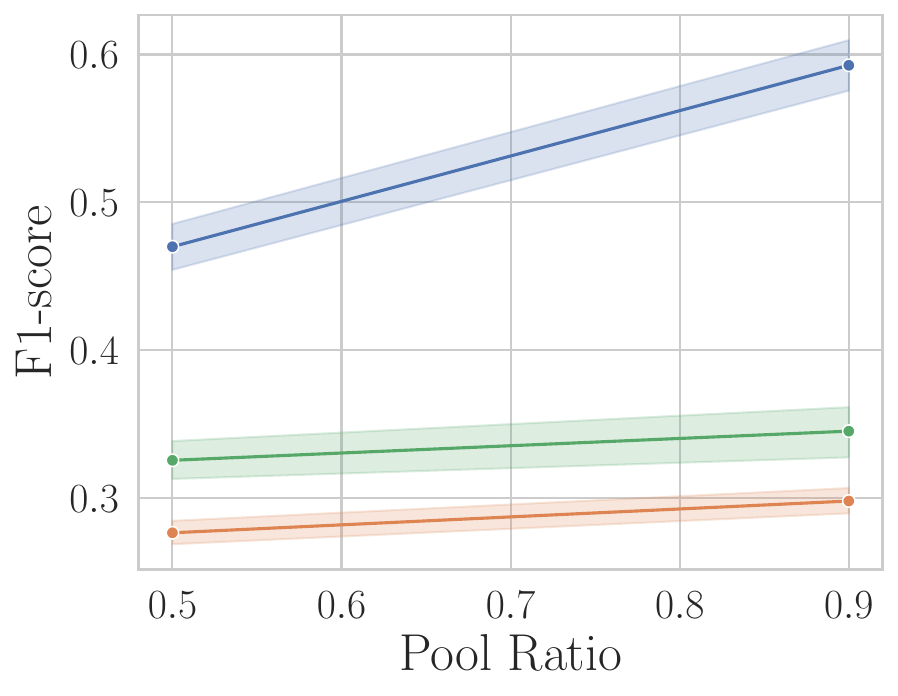}
    \caption{\texttt{hospital-icu}}
    \end{subfigure}
    \begin{subfigure}{\textwidth}
    \centering
    \includegraphics[width=.28\textwidth]{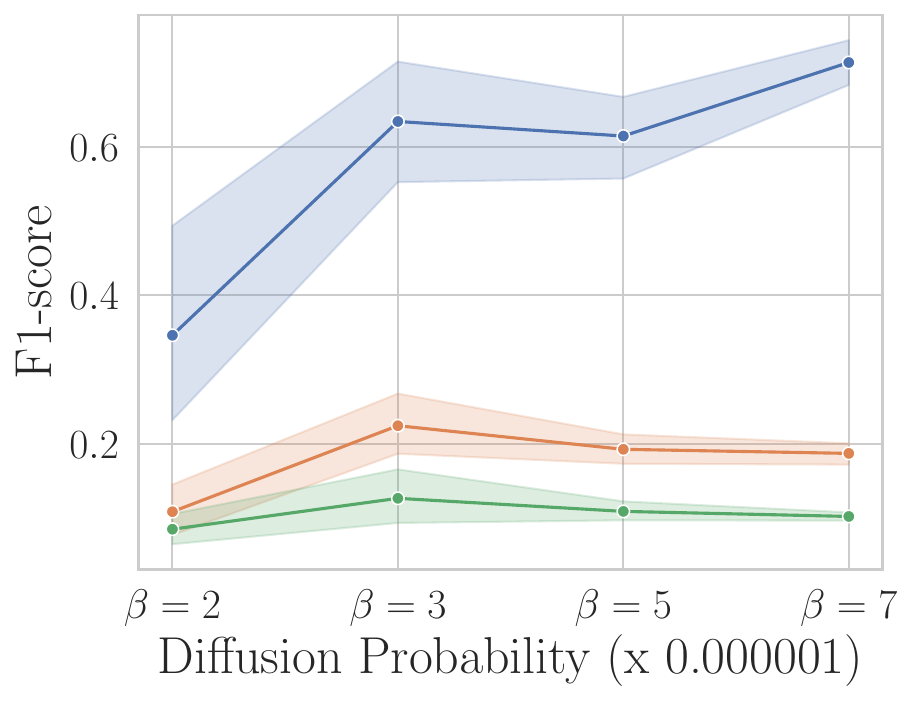}
    \includegraphics[width=.28\textwidth]{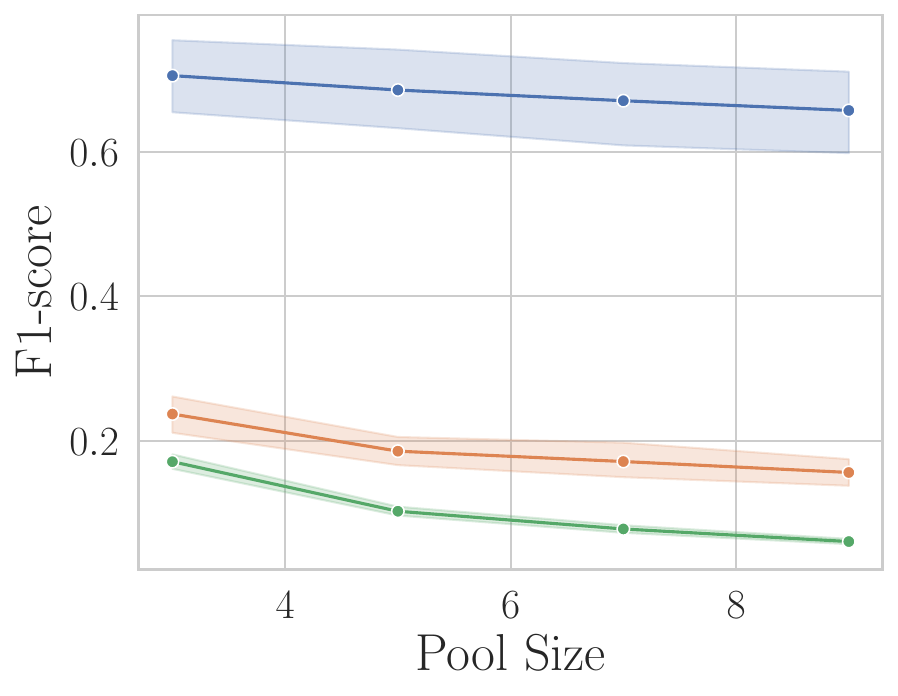}
    \includegraphics[width=.28\textwidth]{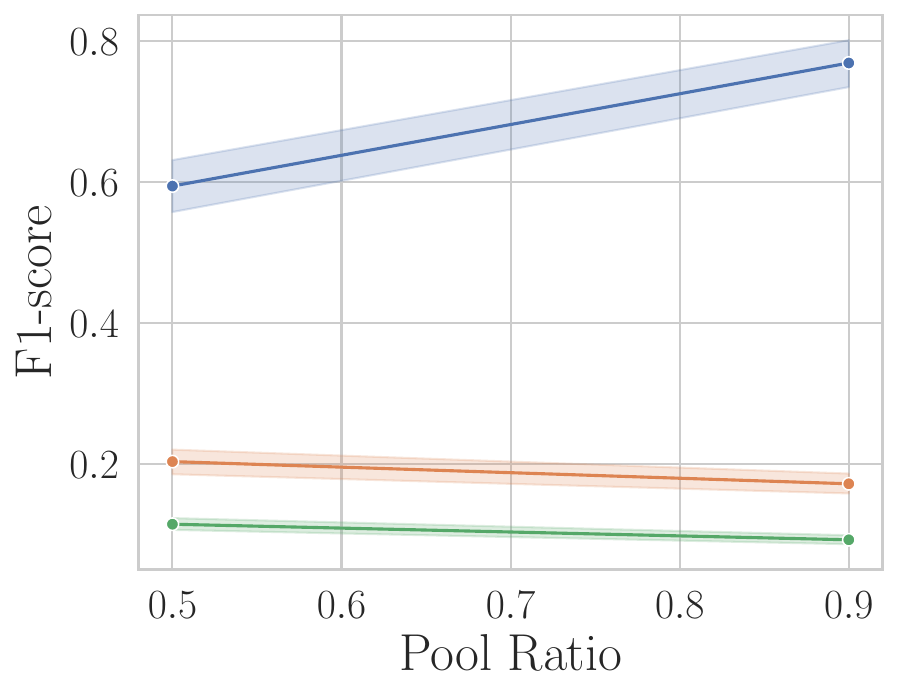}
    \caption{\texttt{small-city}}
    \end{subfigure}
    \includegraphics[width=0.6\textwidth]{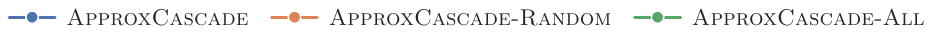}
    \caption{Performance of \textsc{ApproxCascade} against baselines.
\label{fig:performance}
}
\end{figure*}

\textbf{\algo{}}

\textit{Effect of network structure and diffusion probability. } In Figure~\ref{fig:performance}, we report the F1-scores for \textsc{ApproxCascade} and the baselines for different networks.
Across networks, we observe that \textsc{ApproxCascade} vastly outperforms the baselines under the low diffusion probability conditions. The performance gap is highest on the large weighted network \texttt{small-city}. However, in high probability regimes, the performance suffers especially against \textsc{ApproxCascade-All}. This could be explained by the fact that when the underlying cascade is small, only a small subset of nodes in a positive pool are actually infected. In this scenario, choosing infections carefully with \textsc{ApproxCascade} yields dividends. On the other hand, when the cascade is large and almost all nodes in a pool are infected, the \probPoolMLE{} solution under-selects, trying to keep costs low. The exact threshold beyond which \textsc{ApproxCascade-All} is superior is connected to the diffusion probability threshold beyond which large cascades are common. For example, on $G(n=1000, q=0.02)$, this transition is in $(0.05,0.10)$ whereas on BA $(n=1000, m=3)$, it is between $(0.10, 0.20)$.

\textit{Pool size and pool ratio. }
In Figure \ref{fig:performance}, we plot the F1-scores of all the methods versus the pool size. On most networks, \textsc{ApproxCascade} is consistently better than the baselines for all pool sizes. This shows that there is a range of pool sizes for which \textsc{ApproxCascade} yields good solutions. Secondly, increasing the pool
size leads to degradation in the performance. This
is to be expected as the pools get bigger, the number of possible subsets of infected nodes grows exponentially. This makes it likelier to obtain reconstructed nodesets that have less overlap with the ground truth.
There is little impact of the pool ratio on the relative performance of our method against the baselines. The performance improves with higher pool ratio, which is due to the fact that there are fewer unobserved infections outside the pooled set.

\textit{Estimating prevalence. } How well does \algo{} perform as a method for prevalence estimation? Figure \ref{fig:perf-prev} shows the prevalence $e_{rel}$ scores obtained by the methods with respect to the ground truth cascade size. We find that \algo{} outperforms the baseline in nearly all regimes, remaining within $[-0.5,0.5]$ in a wide range of cascade sizes. 
\algo{} tends to underestimate the cascade size which is typical for a MLE cascade method. Like in missing infection recovery, the gap is wider for low diffusion probabilities and narrower for the $G(n,q)$ network. Also, \textsc{ApproxCascade-All} massively overestimates in this task (plots moved to the appendix).

\noindent
\textbf{\algoOne{}}\\
We report the F1-scores obtained by \algoOne{} and the random baseline in  Figure \ref{fig:round_perf} on the missing infections recovery task. \algoOne{} outperforms the baseline across disease regimes and networks. The performance gap is consistent across varying seeding probabilities and edge-diffusion probabilities.

\begin{figure}
    \centering
    \begin{subfigure}{0.48\columnwidth}
        \centering
        \includegraphics[width=\linewidth]{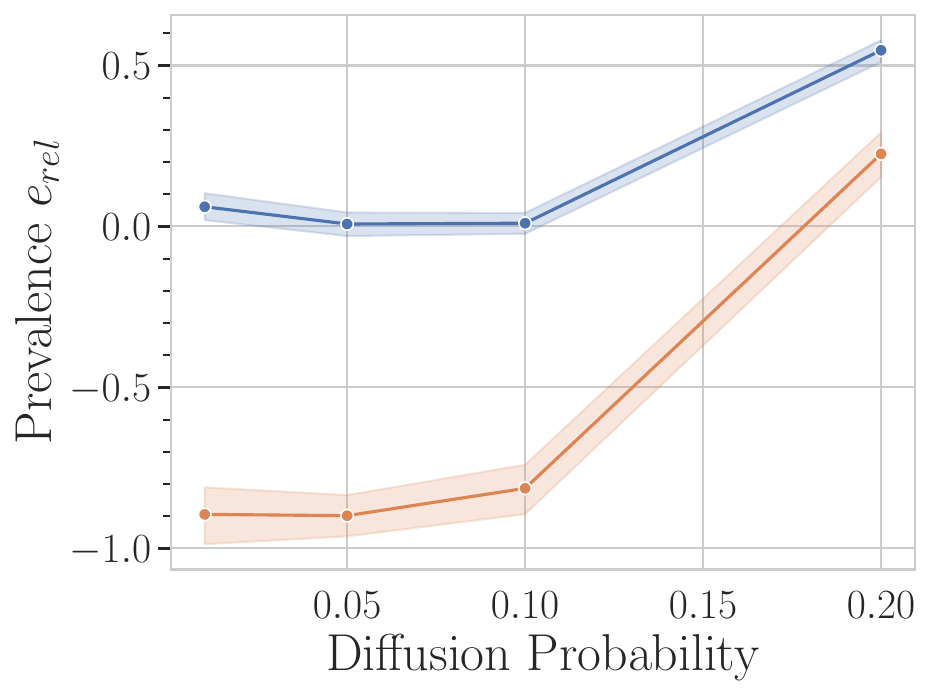}
        \caption{BA $m=3$}
    \end{subfigure}
    \begin{subfigure}{0.48\columnwidth}
        \centering
        \includegraphics[width=\linewidth]{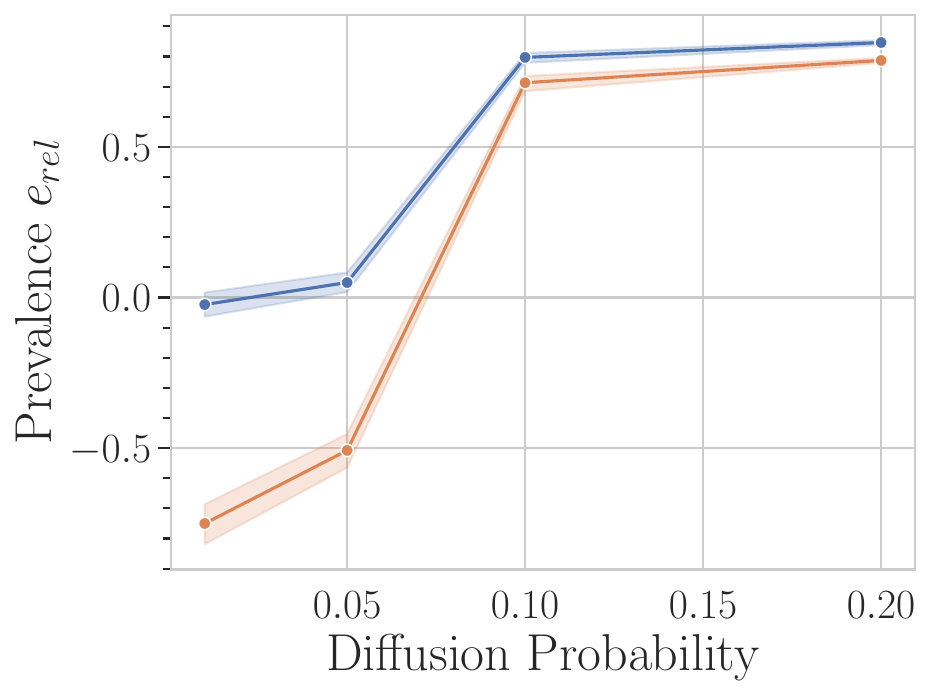}
     \caption{$G(n=1000,q=0.02)$}
    \end{subfigure}
    \begin{subfigure}{0.48\columnwidth}
        \centering
    \includegraphics[width=\linewidth]{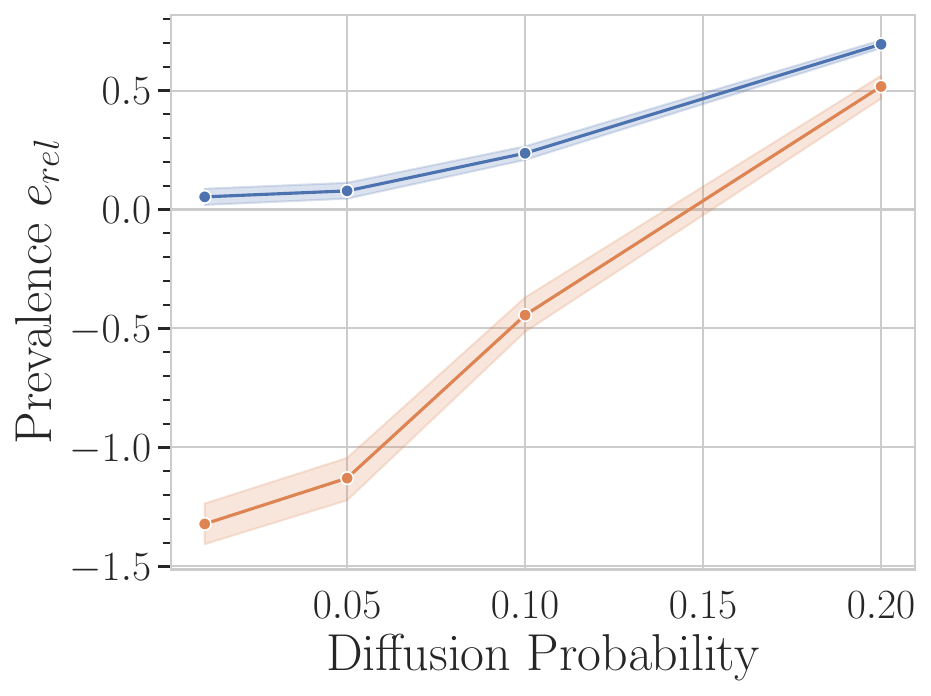}
    \caption{\texttt{hospital-icu}}
    \end{subfigure}
    \begin{subfigure}{0.48\columnwidth}
        \centering
    \includegraphics[width=\linewidth]{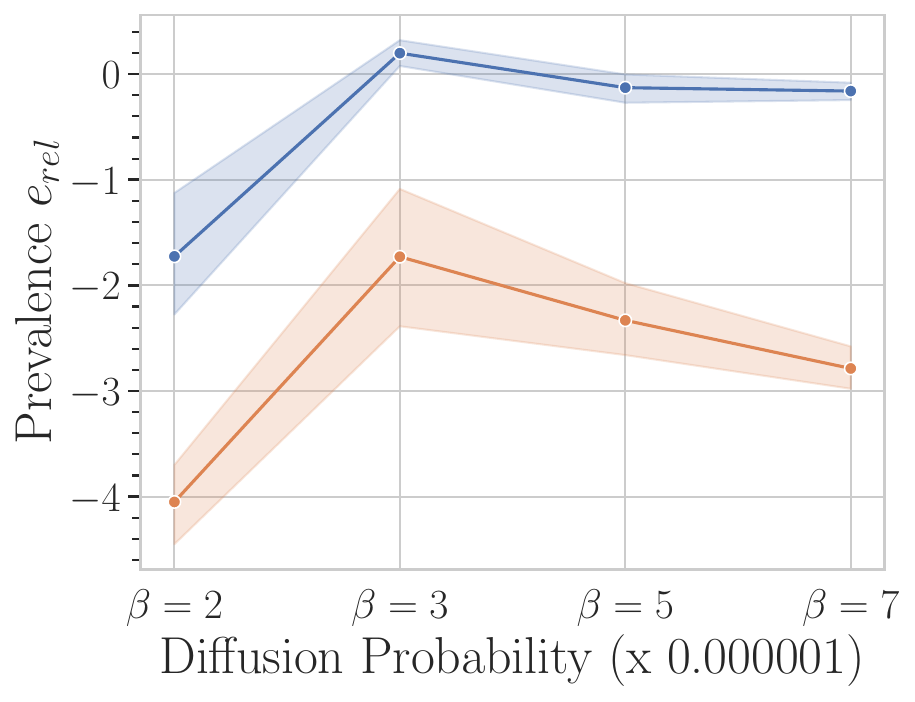}
    \caption{\texttt{small-city}}
    \end{subfigure}
    \includegraphics[width=0.95\linewidth]{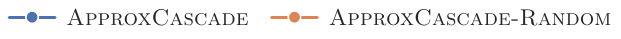}
    \caption{Performance comparison in terms of prevalence estimation relative error $e_{rel}$.}
    \label{fig:perf-prev}
\end{figure}

\begin{figure}
    \centering
    \begin{subfigure}{\columnwidth}
    \centering
    \includegraphics[width=0.48\linewidth]{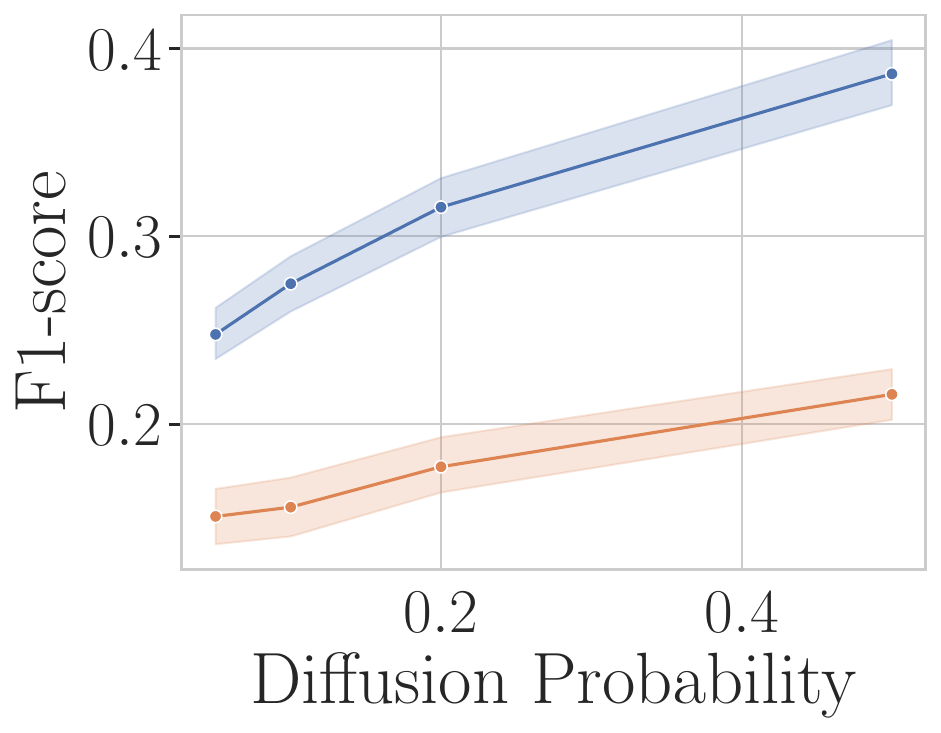}
    \includegraphics[width=0.48\linewidth]{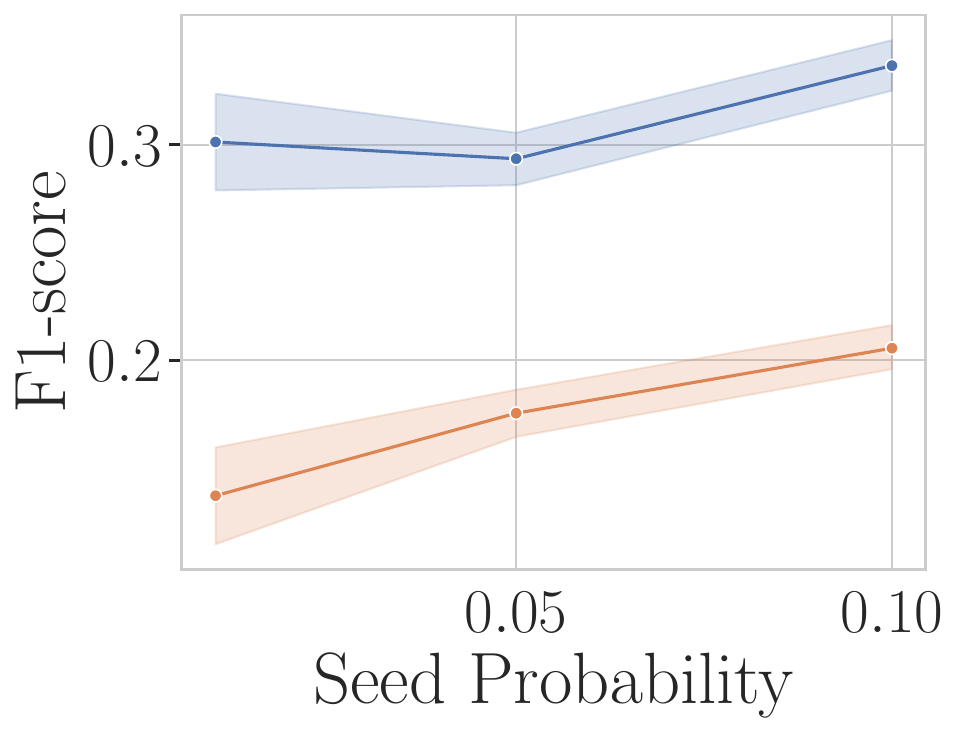}
    \caption{BA $m=3$}
    \end{subfigure}
    \begin{subfigure}{\columnwidth}
        \centering
        \includegraphics[width=0.48\linewidth]{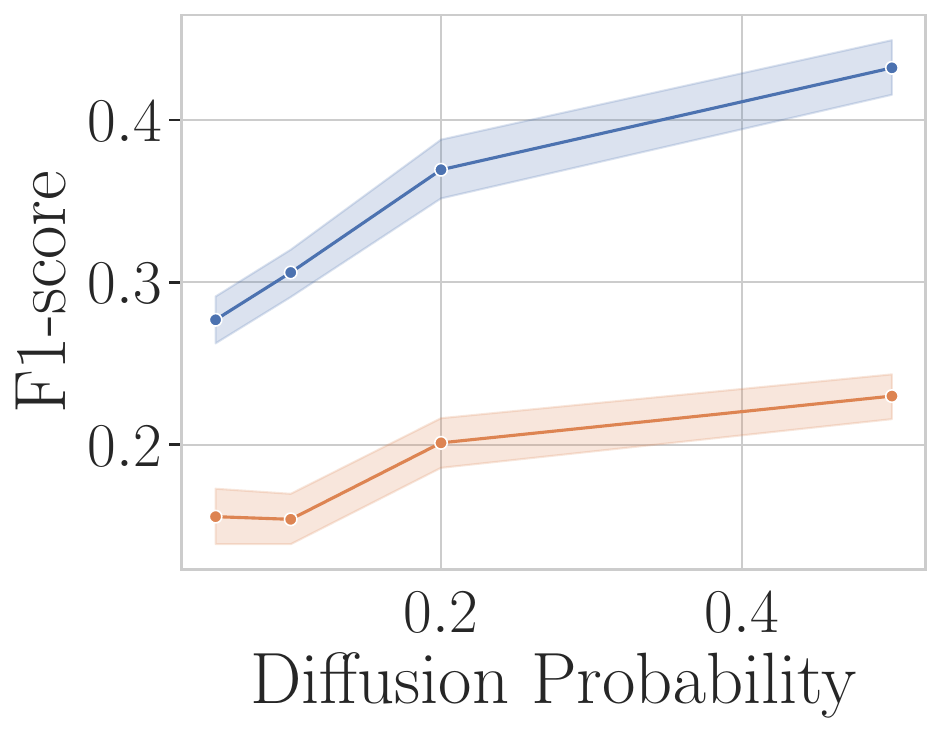}
        \includegraphics[width=0.48\linewidth]{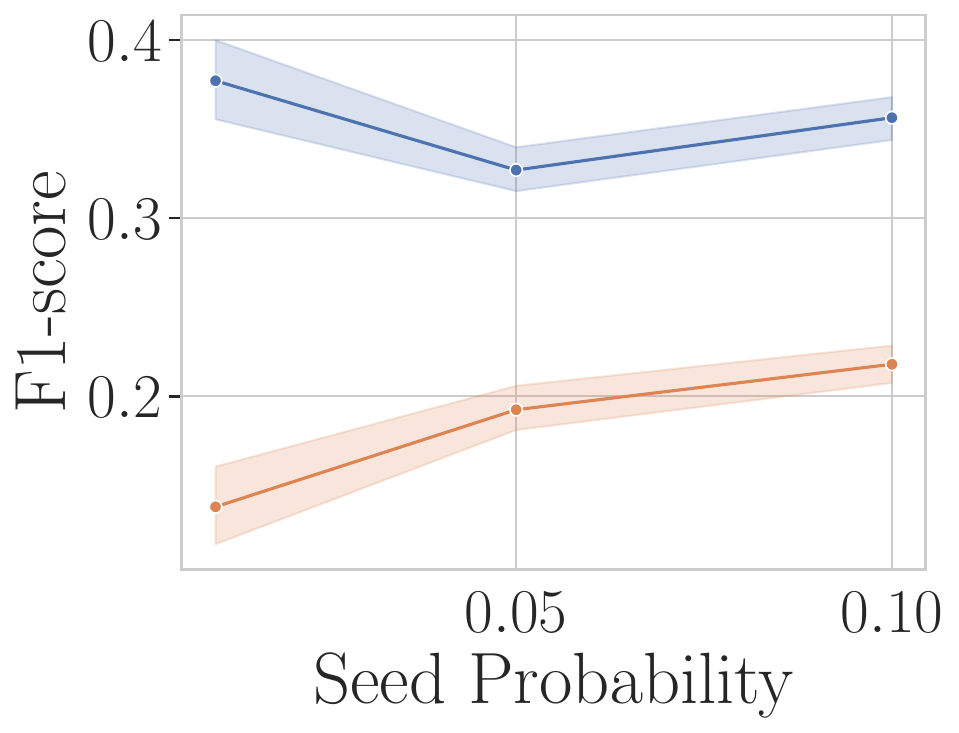}
        \caption{\texttt{hospital-icu}}
    \end{subfigure}
    \includegraphics[width=0.9\linewidth]{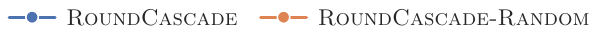}
    \caption{Performance of \textsc{RoundCascade} against the baseline.}
    \label{fig:round_perf}
\end{figure}

\section{Limitations of the MLE cascade approach}
\begin{figure}[!htb]
    \centering
    \includegraphics[width=0.55\columnwidth]{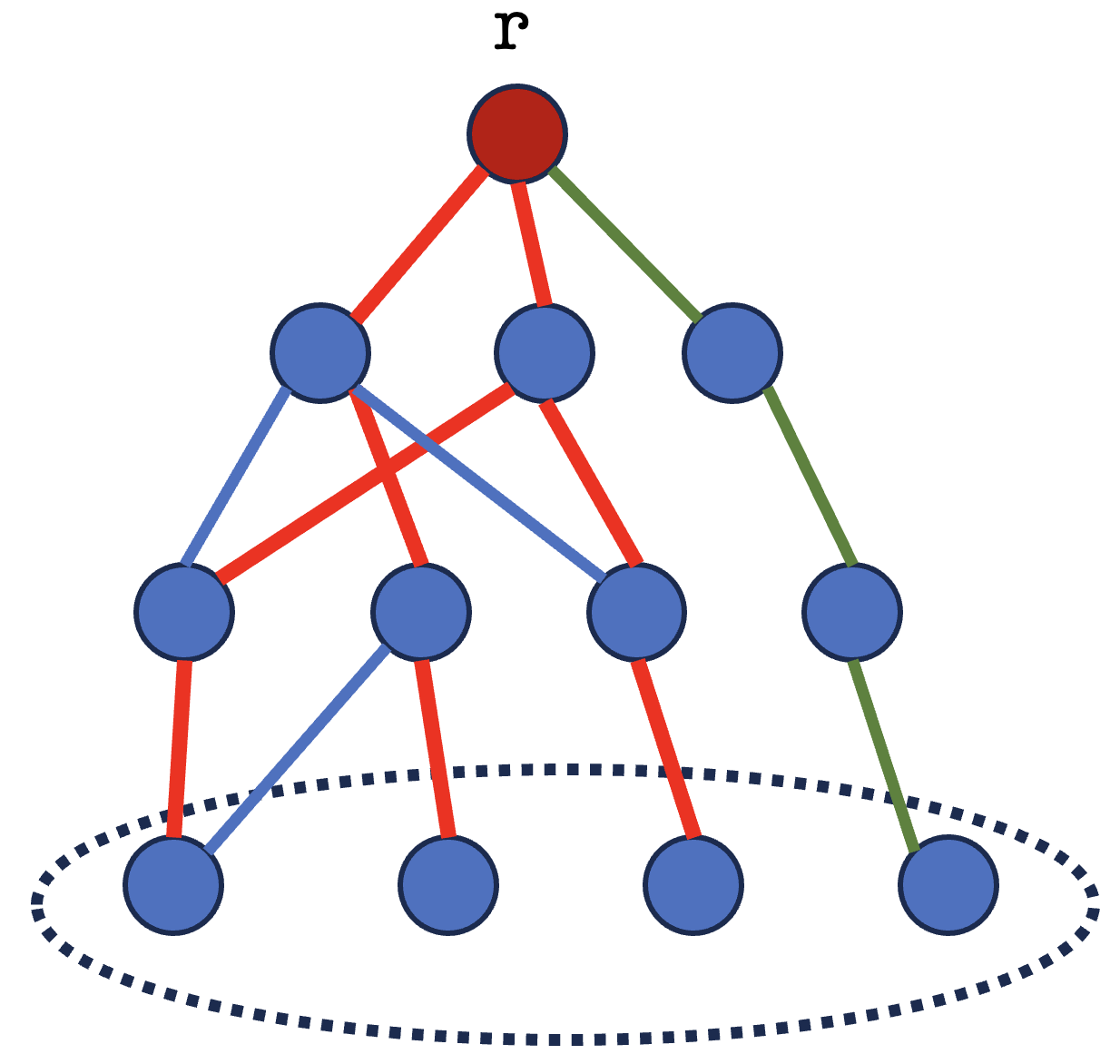}
    \caption{Ground truth cascade shown in red. \probPoolMLE{} solution shown in green. The testing pool is the dotted circle around the leaf nodes.}
    \label{fig:limit_poolMLE}
\end{figure}
\noindent

\begin{observation}
    There exist instances where the \probPoolMLE{} solution does not recover any part of the ground truth cascade.
\end{observation}
\noindent
In Figure \ref{fig:limit_poolMLE}, the  testing pool, shown by the dotted circle around the infected leaf nodes, leads to a positive test. As a \probPoolMLE{} solution needs to connect the root $r$ to at least one node in the pool, while minimizing the \Cost{}, the green path is picked over the higher cost ground truth tree shown in red.
This arises due to the infected leaf nodes being pooled together whereas if they were individually tested, the MLE cascade would be closer to the ground truth.

\noindent
\textbf{Impact of noisy testing.} 
We demonstrate an example
showing that the MLE solution can be very different compared to the ground
truth in a noisy test condition. The example network $G(V,E)$ consists of
three sets of nodes-- $V=\{u\}\cup V'\cup V''$, where~$u$ is a special
\emph{central} node, $V'={v_1,\ldots,v_k}$ such that~$V'$ induces the
path $v_1v_2\cdots v_k$, and $V''=\{w_{ij}\mid i,j=1,\cdots,k\}$ such
that $uw_{i1}w_{i2}\cdots w_{ik}v_i$ is an induced path for
all~$i=1,\cdots,k$. Let each~$P_i=\{v_i,w_{ik}\}$ be a pool. There are
exactly~$k$ pools. Also, let the false positive probability~$q_{fp}=0$,
while the false negative probability~$q_{fn}>0$. Suppose that the cascade
graph is precisely the graph~$v_1v_2\cdots v_k$. Then, all pool tests will
be positive. However, the probability that at least one of the
pools in $\{P_2,P_3,\cdots,P_{k-1}\}$ is reported as negative while~$P_1$
and~$P_k$ are reported as positive
is~$(1-q_{fn})^2\big(1-(1-q_{fn})^{k-2}\big)$. If this ``bad'' event
occurs, then, any MLE solution must contain $v_1$ and $v_k$, which cannot
be connected via the path~$v_1v_2\cdots v_k$. This in turn implies that it
must involve the central node~$w$ and at least~$2k$ nodes from $V''$. Thus
the MLE solution differs from ground truth by more than 50\% of the nodes.
Note that the probability of the bad event can be made arbitrarily high by
increasing~$k$.

\begin{observation}
There exist instances for which the solution to the \probPoolMLEnoisy{} has $o(n)$ overlap with the ground truth, compared to a solution to \probPoolMLE{}.
\end{observation}

The above construction implies that the MLE solution under noise is significantly different from that without noise.
It can be shown that for the above instance, carefully chosen pool tests involving the nodes on the path can ensure that the path is inferred as being infected.
This implies that overlapping tests are more powerful than non-overlapping tests for the epidemic reconstruction problem. This is a significant contrast with the observation in~\cite{finster2023welfaremaximizing} that non-overlapping pools are approximately optimal in the welfare maximization problem.

\section{Conclusion}
We introduce the \probPoolMLE{} problem for reconstructing an epidemic outbreak under group surveillance.
Despite a lot of work in group surveillance, including for infectious diseases, wastewater monitoring, the problem of reconstructing an outbreak hasn't been studied before.
\probPoolMLE{} is NP-hard to even approximate, and we design an approximation algorithm using Group Steiner tree techniques. We consider \probOneHop{}, which is a special case for an IC process that can spread for a maximum of one-hop from the unknown seeds. We find that even this problem is NP-hard to approximate and use randomized rounding on an LP relaxation to achieve a logarithmic approximation bound.
Experiments on synthetic and real contact networks, including a hospital contact network show that our methods which systematically connect infections in a pool considerably outperform baselines which do not do so.
We also show that noise has a very significant impact on the MLE solution.
While our work assumes that the pools are given,
designing optimal pools in the context MLE algorithms for better
inference of the outbreak is an interesting future direction that can
help improve surveillance strategies. 




\begin{acks}
This research is partially supported by NSF grants CCF-1918656 and CNS-2317193, and DTRA award HDTRA1-24-R-0028, Cooperative Agreement number 6NU50CK000555-03-01 from the Centers for Disease Control and Prevention (CDC) and DCLS, Network Models of Food Systems and their Application to Invasive Species Spread, grant no. 2019-67021-29933 from the USDA National Institute of Food and Agriculture, Agricultural AI for Transforming Workforce and Decision Support (AgAID) grant no. 2021-67021-35344 from the USDA National Institute of Food and Agriculture.
\end{acks}



\bibliographystyle{ACM-Reference-Format} 
\bibliography{references}

\appendix
\onecolumn
\begin{center}
\fbox{{\Large\textbf{Technical Supplement: Reconstructing Network Outbreaks under Group Surveillance}}}
\end{center}

\section{Additional details for Section 4}

\subsection{Proof of Theorem~\ref{thm:poolmlehardness}}
This follows from Lemma \ref{lem:pool2gst} and the hardness of the group Steiner tree problem~\cite{halperin2003polylogarithmic}.


\subsection{Proof of Theorem~\ref{thm:onehophardness}}
 We use a reduction from the Minimum Set Cover (MSC) problem which is defined as follows. We are given
     \begin{itemize}
         \item The set of all elements $U=\{e_1, e_2, \dots e_n\}$.
         \item The collection of subsets $S=\{S_1,S_2, \dots, S_m\}, S_i \subseteq U$.
     \end{itemize}
     Given an instance of the MSC problem, we will construct an instance of \probOneHop{}.
     \begin{enumerate}
         \item Construct a bipartite graph $(V_0,V_1,E)$ where we add a node to $V_1$ for each element in $U$ and a node to $V_0$ for each set in $S$. We add an edge $(i,j), i\in V_0, j\in V_1$ if and only if $e_j\in S_i$.
         \item For each node in $V_0$, the homogenous seeding probability $p^0$ is set such that seeding cost $a \gg b$, the cost of not seeding. For each node in $V_1$, seeding probability is zero.
         \item For each edge, the homogenous transmission probability is set such that $c \gg d$.
       
         \item We set the observations $\Gamma_0=\emptyset$, and $\Gamma_1 = V_1$.
   
     \end{enumerate}
     This construction can be done in $O(|U|\times |S|)$ time. Any consistent cascade must cover all nodes in $V_1$. The choice is in picking nodes in $V_0$. 
     Let us suppose we are given a solution to the constructed instance of \probOneHop{}, cascade $c_A$ which includes a subset $A\in V_0$ and~$V_1$ as nodes. $c_A$ includes exactly $|V_1|$ edges, one edge per node in $V_1$ due to $c>d$. Let $\delta_A$ be the cutset of edges from $A$ to $V_1$.
     \begin{align*}
         \Cost^1(c_A) = a|A| + b(|V_0|-|A|) + c|V_1| + d(|\delta_A| - |V_1|)
     \end{align*}

     Consider an optimal solution for the \probOneHop{}, $c^*$ with $A^*\subseteq V_0$ as the seed set. We will show that~$A^*$ corresponds to a minimum set cover for the MSC instance. Firstly, since $V_1$ is the
     observation set, the sets in the MSC instance corresponding to~$A^*$ constitute a set cover.
     Suppose there exists a set~$A'\subset V_0$ such that~$|A'|<|A^*|$ and~$c_{A'}$ is a consistent cascade. Then, $\Cost^1(c^*)-\Cost^1(c_{A'})=(a-b)(|A^*|-|A'|) + d(|\delta_{A^*}|-|\delta_{A'}|)$.
Setting $a-b>d|V_0||V_1|$, we note that~$\Cost^1(c_{A'})<\Cost^1(c^*)$, a 
contradiction to the fact that $c^*$ is an optimal cascade. Hence, $A^*$ corresponds to a minimum set cover.

To prove approximation hardness, we note that given any minimum set cover
$A$, the cost of the corresponding cascade $c_{A}$ is
\[
a|A|+c|V_1|+b(|V_0|-|A|)\le\Cost^1(c_A)\le a|A|+c|V_1|+b(|V_0|-|A|) + d|V_0||V_1|\,.
\]
By setting $a-b\gg d|V_0||V_1|$, we can achieve $\Cost^1(c^*)\le\Cost^1(c_A)\le\Cost^1(c^*)(1+\epsilon)$ for
some small constant~$\epsilon$ where~$c^*$ is an optimal cascade. This ensures
that the $\Omega(\log n)$ approximation hardness  of the MSC problem translates to the \probOneHop{} problem as well.

\section{Additional details for Section 5.1}
\textbf{Reduction in Lemma \ref{lem:nwgst2ewdst}. }
Denote an edge from $u$ to $v$ with weight $w(u,v)$ as a 3-tuple $(u,v,w(u,v))$, a node $v$ with weight $w(v)$ as $(v, w(v))$. Define utility functions AddEdge$(G, (u, v, w(u,v)))$ as adding edge $(u,v,w(u,v))$ to graph $G$, and AddNode$(G, (v,w(v)))$ as adding node $(v,w(v))$ to graph $G$. The reduction is described in Algorithm~\ref{alg:GST}.
\begin{algorithm}
    \caption{(\citet{charikar1999})~GroupSteinerTree$(G, r, \Gamma_1)$}
    \label{alg:GST}
    \raggedright{
    \textbf{Input: } Undirected graph $G=(V,E)$ with node- and edge-weights, seed $r$, Group Terminals $\Gamma_1$ \\
    \textbf{Output:} Group Steiner tree $T_r (\Gamma_1)$
    }
    \begin{algorithmic}[1]
    \State Initialize empty digraph $G'$
    
    \For {each $(u, w(u) \in V$}
        \State AddNode$(G', (u_{in}, w(u))$
        \State AddNode$(G', (u_{out}, w(u))$
        \State AddEdge$(G', (u_{in}, u_{out}, w(u)))$
    \EndFor
    \For{each edge $(u,v) \in E$}
        \State AddEdge$(G', (u_{out},v_{in}, w(u,v) ))$
        \State AddEdge$(G', (v_{out}, u_{in}, w(u,v)))$
    \EndFor
    \For{each group $g \in \Gamma_1$}
        \State AddNode$(G', d_g)$
        \For{each node $v_{out}\in g$}
            \State AddEdge$(G', (v_{out}, d_g, 0)$
        \EndFor
    \EndFor
    \State Let $S = \{d_g\}_{g\in \Gamma_1}$
    \State Solve $T_r =$ \textsc{DirectedSteinerTree}$(G, r, S)$
    \State \textbf{return} $T_r$
    \end{algorithmic}

\end{algorithm}

\textbf{Theorem \ref{theorem:approx_bound}. }
    Let $\hat{T}_r$, rooted at $r$, be the tree returned by the Algorithm \ref{alg:approx_cascade}.
    Let $T^*_{r^*}$ be an optimal solution to \probPoolMLE{}, rooted at $r^*$. Then,
    \begin{align*}
        \Cost(\hat{T}_r) \leq O(k^\epsilon)\Cost(T^*_{r^*}),
    \end{align*}
    where $k$ is the number of positive pools, i.e., $|\Gamma_1|$ and $\epsilon>0$.

\begin{proof}
    Algorithm \ref{alg:approx_cascade} finds the approximate Group Steiner tree $\hat{T}_v$ rooted at $v$ for every node $v \in V(G')$ and returns a tree $\hat{T}_r$ which minimizes the \Cost{}. Thus,
    \begin{align}
        \Cost(\hat{T}_r) \leq \Cost(\hat{T}_v) \label{eq:hatBest}
    \end{align}
    for any node $v$.
    
    From Lemma \ref{lem:pool2gst}, we have,
    \begin{align}
        & \Cost(\hat{T}_r) \leq w(\hat{T}_r) \leq 2\ \Cost(\hat{T}_r)\ \text{and} \label{eq:hatCost} \\
        & \Cost(T^*_{r^*}) \leq w(T^*_{r^*}) \leq 2~\Cost(T^*_{r^*}) \label{eq:starCost}
    \end{align}

Let $T^{gst}_{r^*}$ be an optimal Group Steiner tree on $G'$, rooted at $r^*$. Due to Lemma \ref{lem:nwgst2ewdst}, and the $O(k^\epsilon)$-approximation bound from \cite{charikar1999},
\begin{align}
\label{eq:hatGst}
    w(\hat{T}_{r^*}) \leq O(k^\epsilon) w(T^{gst}_{r^*})
\end{align}
Since $T^*_{r^*}$ is also a group Steiner tree on $G'$, rooted at $r^*$, its weight is lower bounded by that of $T^{gst}_{r^*}$.
\begin{align}
    \label{eq:gstTstar}
    w(T^{gst}_{r^*}) \leq w(T^*_{r^*})
\end{align}

From (\ref{eq:hatBest}),(\ref{eq:hatCost}),(\ref{eq:starCost}),(\ref{eq:hatGst}) and (\ref{eq:gstTstar}),
\begin{align*}
    & \Cost(\hat{T}_r) \leq \Cost(\hat{T}_{r^*}) \leq w(\hat{T}_{r^*}) \leq O(k^\epsilon)w(T^{gst}_{r^*})  \\
    & \quad \leq O(k^\epsilon) w(T^*_{r^*}) \leq 2O(k^\epsilon) \Cost(T^*_{r^*}) \\
    & \quad \Rightarrow \Cost(\hat{T}_r) \leq O(k^\epsilon) \Cost(T^*_{r^*})
\end{align*}
\end{proof}

\subsection{Extension to the {\probPoolMLEnoisy{}} problem}

Given observed pooled test outcomes $(\Gamma_0, \Gamma_1)$, we would like to find a test outcome $(\Gamma_0', \Gamma_1')$ and a tree consistent with the test outcome which maximizes the following probability,
\begin{align*}
    P(\Gamma_0',\Gamma_1'|\Gamma_0, \Gamma_1)P(T_r| \Gamma_0',\Gamma_1')
\end{align*}
Let $q(\Gamma_0,\Gamma_1,\Gamma_0',\Gamma_1') $ denote the probability that the actual outcome is $\Gamma_0',\Gamma_1'$ given that the observed outcomes are $\Gamma_0,\Gamma_1$. This is computed by product of probabilities of false positive and false negative. For instance, if the observed test outcomes for 4 pools were ``1100" and the actual test outcomes were ``1001",
then $q = (1-p_{fp})p_{fp}(1-p_{fn})p_{fn}$.
In terms of costs, we would like to minimize,
\begin{align*}
    \Cost_{noisy}(\Gamma_0',\Gamma_1') = 
    \Cost(T_r) + \log{(1/q(\Gamma_0,\Gamma_1,\Gamma_0',\Gamma_1'))}
\end{align*}

\begin{algorithm}
\caption{\textsc{ApproxCascade-Noisy}}
\raggedright{
\textbf{Input:} A contact network $G=(V,E)$, pool-tested node groups $\Gamma_0$, $\Gamma_1$.\\
\textbf{Output:} Test outcomes $(\Gamma_0', \Gamma_1')$ and a consistent tree $T'$.
}
\label{alg:noisy_approx_cascade}
\begin{algorithmic}[1]
    \State Let $\mathcal{S}$ be the set of all possible test outcomes for the pools in $\Gamma_0 \cup \Gamma_1$.
    \For{$(\overline{\Gamma_0},\overline{\Gamma_1})$ in $\mathcal{S}$}
        \State Compute $q(\Gamma_0,\Gamma_1,\overline{\Gamma_0},\overline{\Gamma_1})$.
        \State $\overline{T_r}$ =  \textsc{ApproxCascade}(G, $(\overline{\Gamma_0},\overline{\Gamma_1})$)
        \State Compute the \probPoolMLEnoisy{} cost.
        \begin{align*}
            \Cost_{noisy}(\overline{\Gamma_0},\overline{\Gamma_1}) = \Cost(\overline{T_r}) + \log{(1/q(\Gamma_0, \Gamma_1,\overline{\Gamma_0},\overline{\Gamma_1}))}
        \end{align*}
    \EndFor
    \State Choose $\Gamma_0',\Gamma_1',T_r'$ from among $S$ which minimizes the $\Cost_{noisy}$.
    \State \textbf{return} $\Gamma_0',\Gamma_1',T_r'$.
    
\end{algorithmic}

\end{algorithm}

Algorithm $\ref{alg:noisy_approx_cascade}$ calls \textsc{ApproxCascade} $2^{|\Gamma|}$ times, which can be sped up by pruning.

\section{Additional details for Section 5.2}

\textbf{Lemma 4: }
Let  $x, y, z$ denote the optimal integral solutions to the above program.
Then, $\sum_{i\in V_0} ( a_ix_i + b_i(1-x_i)) + \sum_{(i,j)\in E} ( c_{ij}y_{ij} + d_{ij}(z_{ij}) )\leq 2OPT$, where $OPT$ denotes the cost of the optimal solution to the instance of \probOneHop{}.

\begin{proof}
For each $ij$, we define $z_{ij}$ as:
(1) if $y_{ij}=1$ then $z_{ij}=0$
(2) if $x_i=0$ then $z_{ij}=0$
(3) if $x_i=1, y_{ij}=0$ then $z_{ij}=1$.
Then $x, y, z$ corresponds to the optimal solution

Given $x, y, z$, construct $x', y', z'$ in the following manner:
$x'=x, y'=y$.
For each $ij$, set $z'_{ij} = x_i$.
We argue that $\sum_{ij} d_{ij} z'_{ij}\leq \sum_{ij} c_{ij} y_{ij} + d_{ij}z_{ij}$, and the statement follows.
Consider any $i,j$.
Suppose $x'_i=0$. Then $z'_{ij}=0$.
Suppose $x'_i=1$. 
Then $z'_{ij}=1$.
If $y_{ij}=0$ then we must have $z_{ij}=1$, and $z'_{ij}\leq z_{ij}$.
If $y_{ij}=1$, we have $d_{ij}z'_{ij}\leq c_{ij}y_{ij}$
\end{proof}

\section{Additional Details for Section~\ref{sec:exp}}

\paragraph{Computation time.} The experiments were conducted on a HPC
cluster where each node consists of two Intel Xeon Gold 6148 CPUs with
20 cores each. The RAM per node is 384GB. The algorithms were
implemented in Python~3.8. The results of computation time are shown
in Figure~\ref{fig:timing}. One of the most significant
factors is the number of positive pools, which determines the 
number of terminals that are input to the Steiner tree algorithm.
An interesting aspect is that the time taken is also significantly
affected by the pool size.

\begin{figure}[!htb]
    \centering
    \includegraphics[width=.45\textwidth]{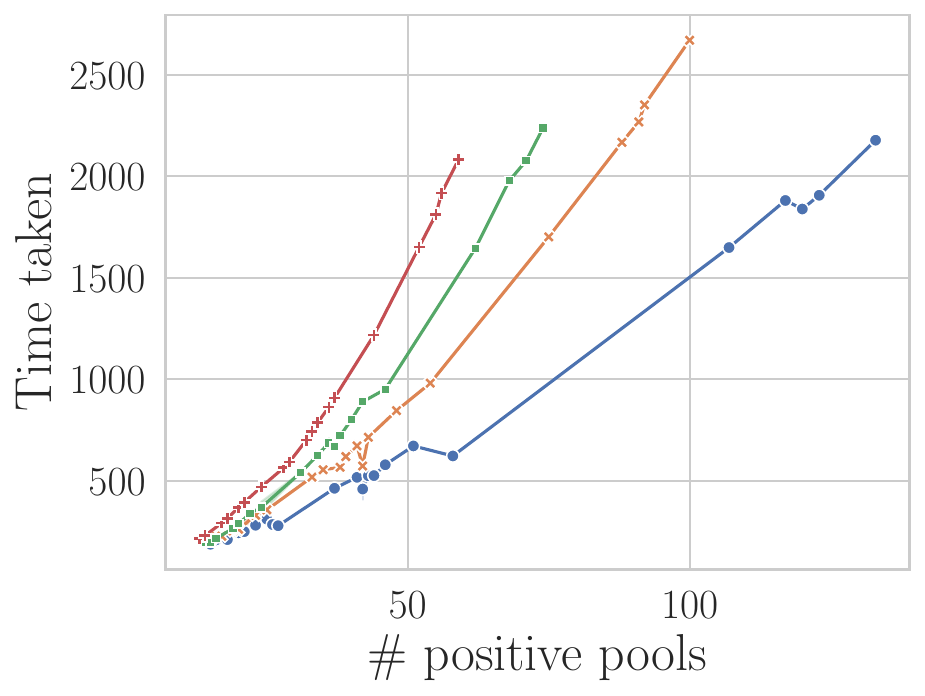}
    \includegraphics[width=0.45\textwidth]{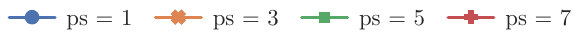}
    \caption{Time taken (in sec) by \textsc{ApproxCascade} on~\texttt{hospital-icu} against number of positive pools. Each line is for a pool size (ps). We fix $p=0.1$ and pooled ratio = 0.5. 
\label{fig:timing}
}
\end{figure}

\begin{figure}
    \centering
    \begin{subfigure}{\columnwidth}
        \centering
            \includegraphics[width=0.3\linewidth]{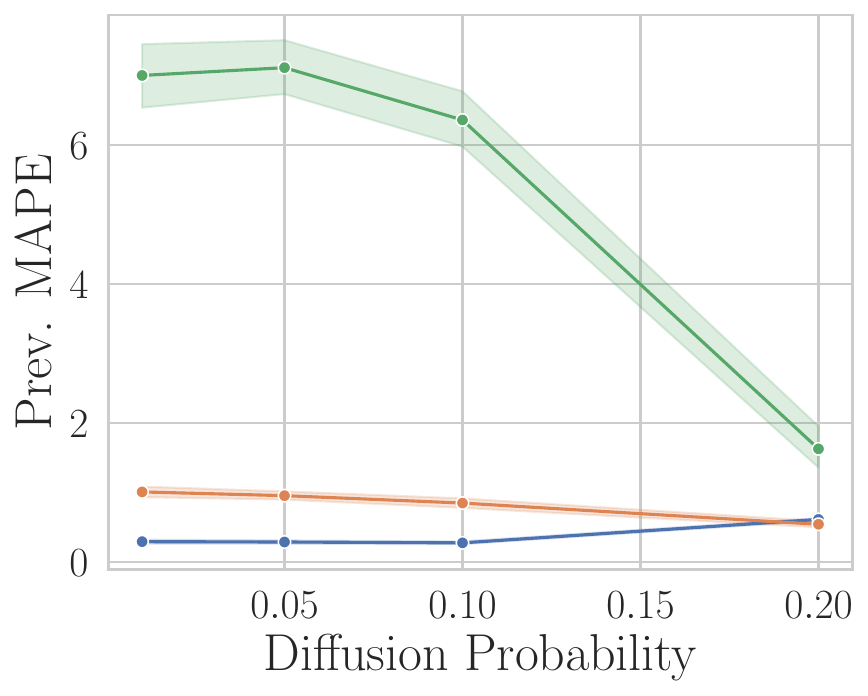}
            \includegraphics[width=0.3\linewidth]{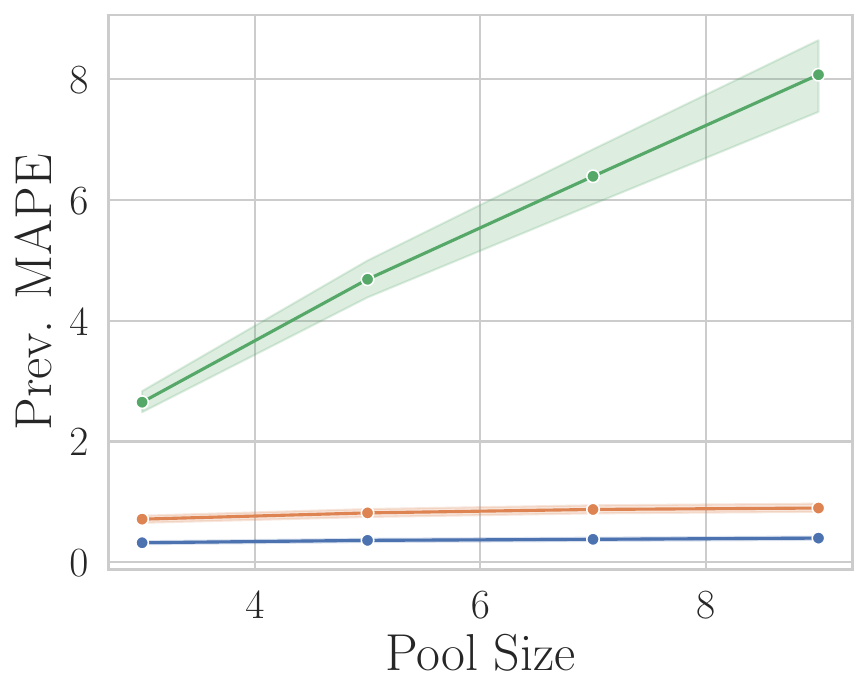}
        \caption{BA $(m=3)$}
    \end{subfigure}
    \begin{subfigure}{\columnwidth}
        \centering
            \includegraphics[width=0.3\linewidth]{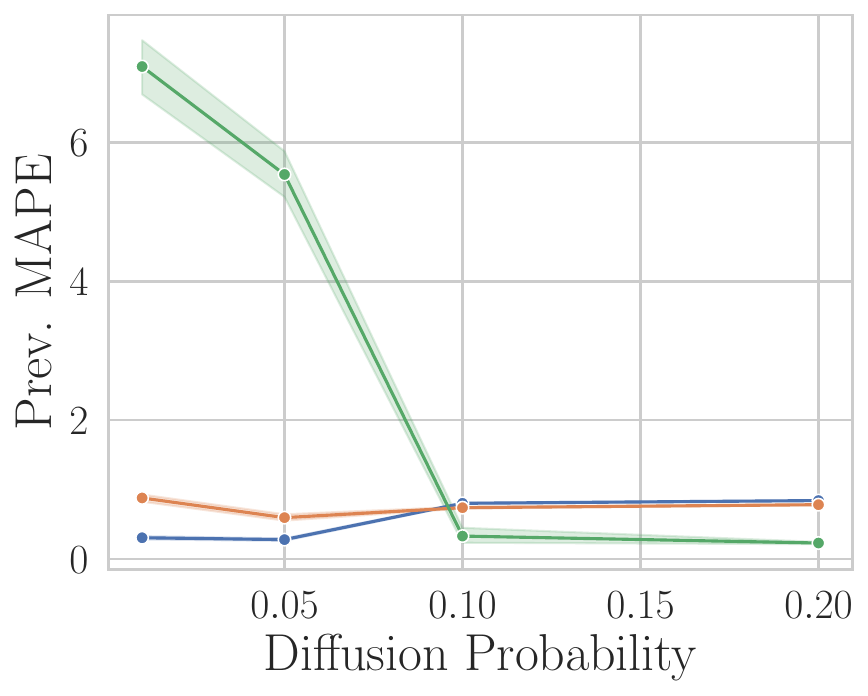}
            \includegraphics[width=0.3\linewidth]{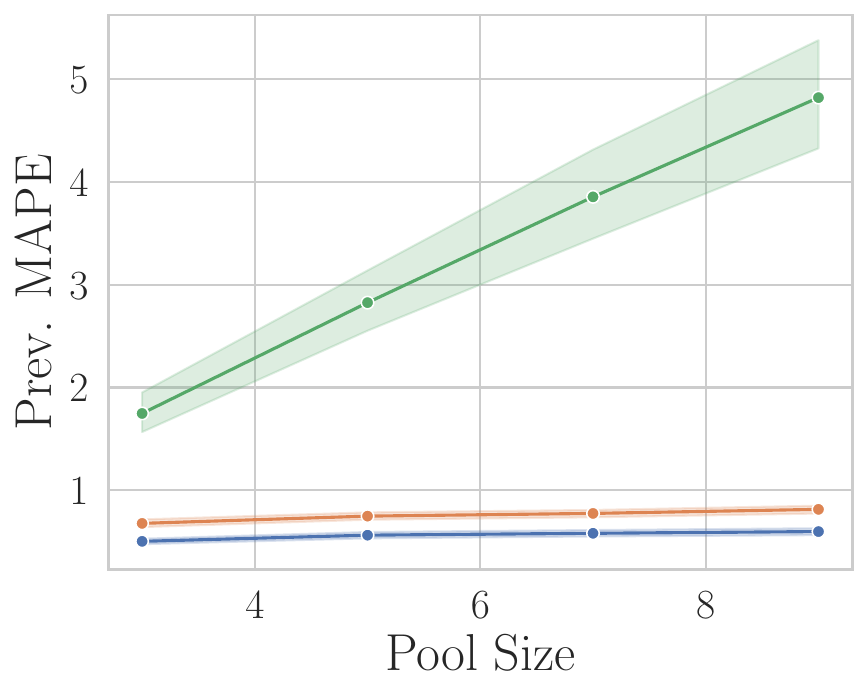}
        \caption{$G(n=1000, q=0.02)$}
    \end{subfigure}
    \begin{subfigure}{\columnwidth}
        \centering
            \includegraphics[width=0.3\linewidth]{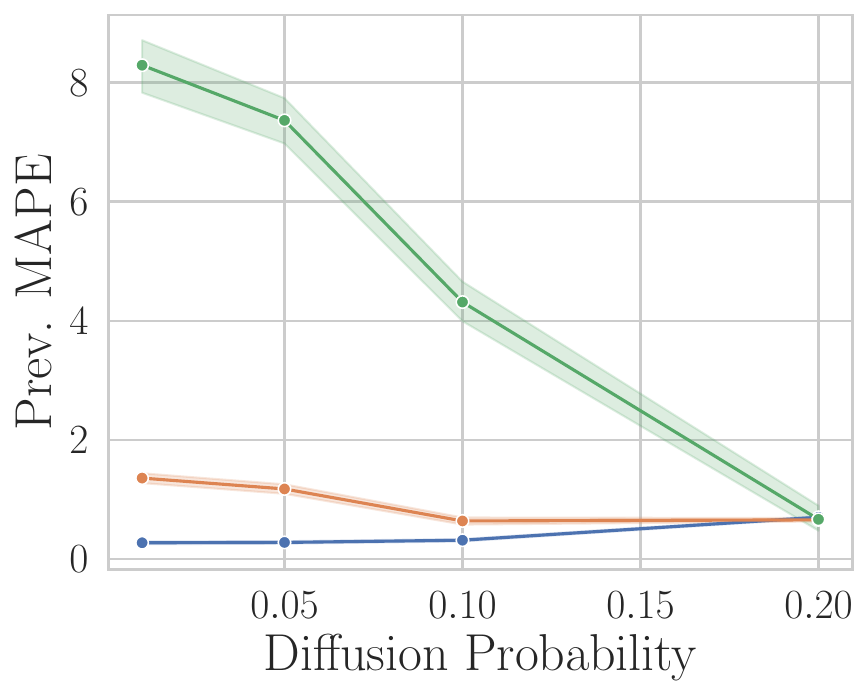}
            \includegraphics[width=0.3\linewidth]{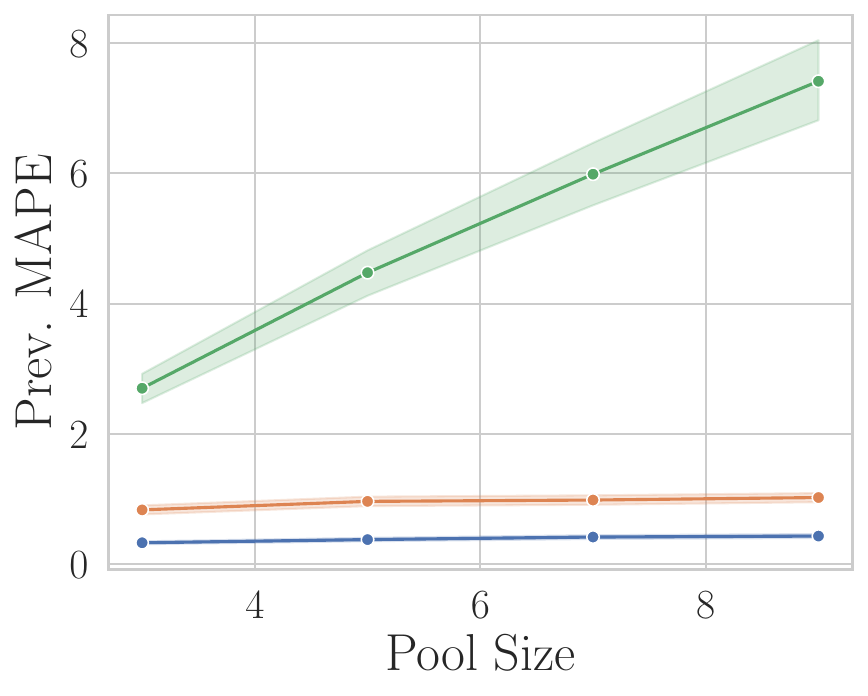}
        \caption{\texttt{hospital-icu}}
    \end{subfigure}
    \begin{subfigure}{\columnwidth}
        \centering
            \includegraphics[width=0.3\linewidth]{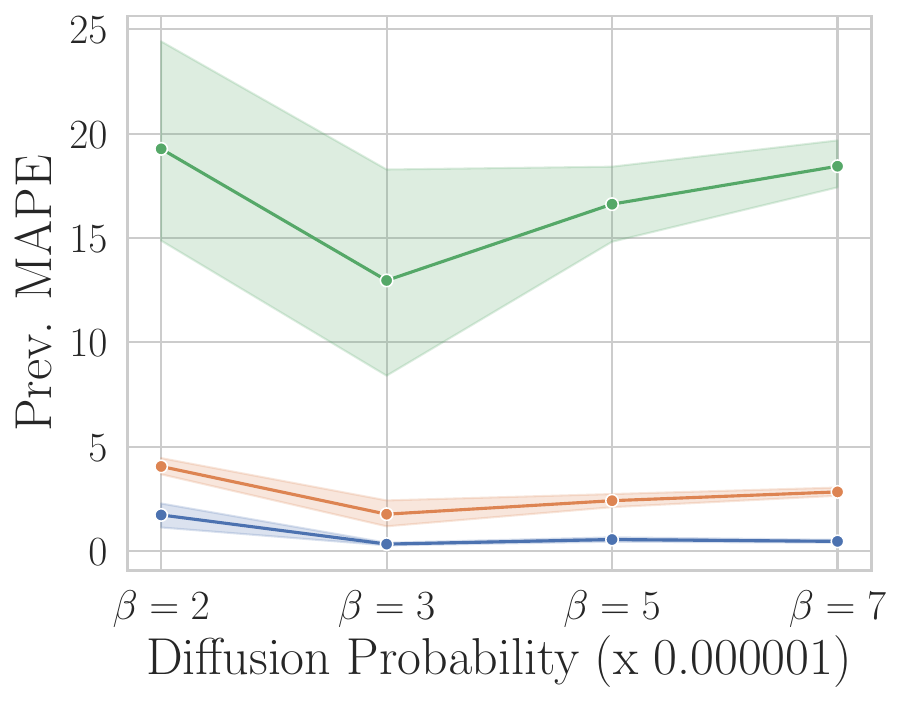}
            \includegraphics[width=0.3\linewidth]{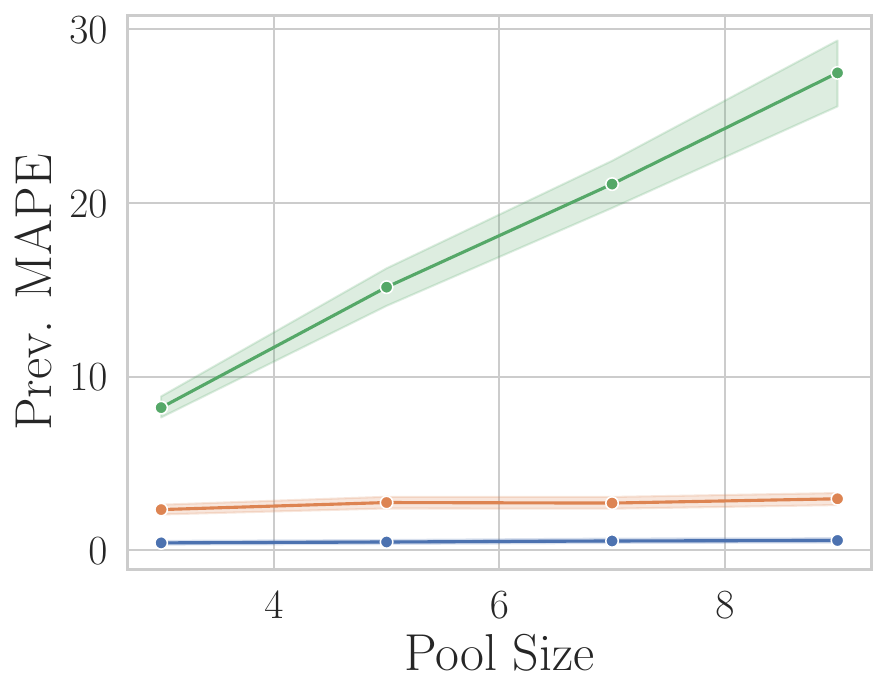}
        \caption{\texttt{small-city}}
    \end{subfigure}

    \caption{Performance of \algo{} against baselines in terms of prevalence estimation error}
    \label{fig:prev_addn}
\end{figure}


\end{document}